\documentclass[letterpaper,12pt]{amsart}
\usepackage[usenames]{color}
\usepackage{xifthen}
\newboolean{colour}
\setboolean{colour}{true}
\title[k-mer methods for phylogenetic trees]{Statistically-Consistent k-mer Methods for Phylogenetic Tree Reconstruction}

\author{Elizabeth S. Allman}
\address{Department of Mathematics and Statistics\\
University of Alaska Fairbanks, 99775}
\email{e.allman@alaska.edu}

\author{John A. Rhodes}
\address{Department of Mathematics and Statistics\\
University of Alaska Fairbanks, 99775}
\email{j.rhodes@alaska.edu}

\author{Seth Sullivant}
\address{Department of Mathematics \\
North Carolina State University, Raleigh, NC, 27695}
\email{smsulli2@ncsu.edu}

\date{}

\usepackage{latexsym,array,delarray,amsthm,amssymb,epsfig,verbatim, natbib}%eufrak
\theoremstyle{plain}
\newtheorem{thm}{Theorem}[section]

\newtheorem{prop}[thm]{Proposition}
\newtheorem{cor}[thm]{Corollary}

\theoremstyle{definition}

\theoremstyle{remark}

\numberwithin{figure}{section}

\newcommand{\diag}{\operatorname{diag}}
\newcommand{\tr}{\operatorname{tr}}
\newcommand{\Var}{\operatorname{Var}}
\newcommand{\Cov}{\operatorname{Cov}}
\newcommand{\EE}{\mathbb E}
\newcommand{\Prob}{\operatorname{Prob}}

\newcommand{\inD}[1][\relax]{\def\argone{#1}\def\temprelax{\relax}
  \ifx\argone\temprelax\right.\else\,\middle|#1\right.{}\fi}

\begin{document}
\maketitle

\begin{abstract}
  Frequencies of $k$-mers in sequences are sometimes used as a basis
  for inferring phylogenetic trees without first obtaining a multiple
  sequence alignment.  We show that a standard approach of using the
  squared-Euclidean distance between $k$-mer vectors to approximate a
  tree metric can be statistically inconsistent. To remedy this, we
  derive model-based distance corrections for orthologous sequences
  without gaps, which lead to consistent tree inference.  The
  identifiability of model parameters from $k$-mer frequencies is also
  studied.  Finally, we report simulations showing the corrected
  distance out-performs many other $k$-mer methods, even when
  sequences are generated with an insertion and deletion process.
  These results have implications for multiple sequence alignment as
  well, since $k$-mer methods are usually the first step in
  constructing a guide tree for such algorithms.
\end{abstract}

%%%%%%%%%%%%%%%%%%%%%%%%%%%%%%%%%%
%%%%%%%%%%%%%%%%%%%%%%%%%%%%%%%%%%
%%%%%%%%%%%%%%%%%%%%%%%%%%%%%%%%%%
%%%%%%%%%%%%%%%%%%%%%%%%%%%%%%%%%%

\section{Introduction}

The first step in most approaches to inference of a phylogenetic tree
from sequence data is to construct an alignment of the sequences,
intended to identify orthologous sites.  When many sequences are
considered at once, a full search over all possible sequence
alignments is infeasible, so most algorithms reduce the range of
possible alignments considered by constructing multiple alignments on
subcollections of the sequences and then merging these together, using
heuristic, rather than model-based, schemes.  Deciding which
subcollections of the sequences to align follows a \emph{guide tree},
a rough tree approximating the evolutionary histories of all the
sequences.  This means that sequence alignment and phylogenetic tree
construction are circularly entangled: finding a tree depends on
knowing a multiple sequence alignment, and obtaining a sequence
alignment requires knowing a tree.

To get around this ``chicken-and-egg'' problem of alignment and
phylogeny several methods have been proposed.  The most theoretically
appealing methods are simultaneous alignment and phylogeny algorithms,
built upon statistical models of insertion and deletions (indels) of
bases as well as base substitutions \citep{Thorne1991,Thorne1992}.
Unfortunately, such methods are computationally intensive, and do not
scale well for large phylogenies.  Alternatively, methods have been
developed that iteratively compute alignments and phylogenies many
times, using the output from one procedure as the input to the next
\citep{Liu2009, Liu2012}. These last investigations underscored that
poor alignments can be a significant source of error in trees, and
that better guide trees can lead to better tree inference.

If one is interested primarily in the phylogeny, an alternate strategy
is to develop methods for inferring trees that do not require having a
sequence alignment in hand.  Current fully alignment-free phylogenetic
methods were not developed with stochastic models of sequence
evolution in mind, and are not widely accepted in the phylogenetics
community.  However, the construction of initial guide-trees for
producing alignments generally follows an alignment-free approach.
For example MUSCLE \citep{Edgar2004a,Edgar2004b} uses $k$-mer
distances with UPGMA or Neighbor-Joining to produce guide trees,
whereas Clustal Omega \citep{Sievers2011} uses a low dimensional
geometric embedding based on $k$-mers \citep{Blackshields2010} and
$k$-means or UPGMA as a clustering algorithm. Thus even though tree
inference is typically performed with model-based statistical methods,
the initial step is built on heuristic ideas, with no evolutionary
model in use.

\smallskip

As exemplified by these alignment algorithms, most common
alignment-free methods are based on $k$-mers, contiguous subsequences
of length $k$.  To a sequence of length $n$ for any natural number
$k\le n$ we associate the vector of counts of its distinct $k$-mers.
For a DNA sequence the $k$-mer count vector has $4^k$ entries and sums
to $n-k+1$.  Distance between two sequences might be calculated by
measuring the (squared Euclidean) distance between their (suitably
normalized) $k$-mer count vectors.  In this way one obtains pairwise
distances between all sequences, and can apply a standard
distance-based method (e.g.~Neighbor joining) to construct a
phylogenetic tree.

Such $k$-mer methods are sometimes described as non-parametric, in
that they do not depend on any underlying statistical model describing
the generation of the sequences. For phylogenetic purposes, where an
evolutionary model will be assumed in later stages of an analysis, it
is hard to view this as desirable.  As we will show in Section
\ref{sec:JCcor}, if we do assume that data is produced according to a
standard probabilistic model of sequence evolution, then a naive
$k$-mer method is statistically inconsistent.  That is, over a rather
large range of metric trees, it will not recover the correct tree from
sequence data, even with arbitrarily long sequences.  The statistical
inconsistency of such a $k$-mer method is similar to the ones seen for
parsimony, in the ``Felsenstein zone'' \citep{Felsenstein1978}.

Our main result, presented in Section \ref{sec:formulas}, is the
derivation of a statistically consistent model-based $k$-mer distance
under standard phylogenetic models with no indel process.  It would,
of course, be preferable to work with a model including indels, as
only in that situation is an alignment-free method of real value. At
this time, however, we are only able to offer a reasonable heuristic
extension of our method for sequences evolving with a mild indel
process.  This appears in Section \ref{sec:practical}.  We view this
as only a first step towards developing rigorously-justified
model-based $k$-mer methods for indel models; solid theoretical
development of such methods is a project for the future.

Section \ref{sec:ident} presents more detailed results on
identifiability of model parameters from $k$-mer count vectors. While
one of these plays a role in establishing the results of Section
\ref{sec:formulas}, they are of interest in their own right. Technical
proofs for Sections \ref{sec:formulas} and \ref{sec:ident} are
deferred to the Appendices.

In Section \ref{sec:simulation} we report results from simulation
studies on sequence data generated from models with and without an
indel process, comparing $k$-mer methods with and without the
model-based corrections.  As expected, the $k$-mer methods with the
model-based corrections outperform both the uncorrected $k$-mer
methods and a more traditional distance method based on first
computing pairwise alignments of sequences.  The simulation studies
also illustrate the statistical consistency of the model-based
methods, and the inconsistency of the standard $k$-mer method.

\medskip

\subsection*{Comparison to Prior Work on Alignment-Free Phylogenetic Algorithms}

There have been a number of papers in recent years developing 
alignment-free methods for phylogenetic tree 
reconstruction \citep{Daskalakis2013,LifePrint, Yang2008,Chan2014}
or for clustering metagenomic data \citep{Reinert2009, Smale2014}.
Of these only one \citep{Daskalakis2013} appears to be based on common 
phylogenetic modeling assumptions, but its focus is theory rather than practice.
Others  \citep{Chan2014,Reinert2009} are model-based but the underlying 
model is not evolutionary in nature.
Some are primarily simulations studies of the application of a method
on larger trees than those we focus on here.  

In our simulations, we follow the framework suggested by
\citet{Huelsenbeck1995}, which allows us to graphically display
performance on an important slice of tree space for 4-taxon trees.
One then readily sees the effect on performance of varying branch
length, and the strength of the common ``long branch attraction''
phenomenon.  In comparison, the simulations in \citep{LifePrint,
  Yang2008,Chan2014} use trees that have more leaves but the range of
branch lengths explored is significantly reduced.  We believe
following Hulsenbeck's plan provides more fundamental insights into a
methodology's value.

\citet{Daskalakis2013} derived a statistically consistent
alignment-free method for a model with indels, although it appears to
have not yet been tested, even on simulated data.  Their method is
based on computing the base distribution (i.e., the 1-mer
distribution) in sub-blocks of the sequences, and motivated the
similar approach we take here.  In addition to restricting to 1-mers,
their approach requires \emph{a priori} knowledge of the value of
certain model parameters, e.g., the proportion of gaps in a sequence,
and several parameters defining the base substitution process.  As our
theoretical results involve no indel process and allow arbitrary $k$,
the two works are not directly comparable.  However, we are able to
obtain stronger results on the identifiability of parameters of the
base substitution model, and our simulations show that using $k>1$ can
result in improved performance.

For advancing data analysis, it is highly desirable to develop
theoretically-justified model-based $k$-mer methods that both account
for indels and require few assumptions on model parameters.  Neither
\citet{Daskalakis2013} nor we provide such methods; both of our works
represent first steps, in slightly different directions, but pointing
towards the same goal.

\medskip

%%%%%%%%%%%%%%%%%%%%%%%%%%%%%%%%%%
%%%%%%%%%%%%%%%%%%%%%%%%%%%%%%%%%%
%%%%%%%%%%%%%%%%%%%%%%%%%%%%%%%%%%
%%%%%%%%%%%%%%%%%%%%%%%%%%%%%%%%%%

\section{$k$-mer formulas for indel-free sequences}\label{sec:formulas}

In this section we present formulas for model-based corrections to
distances based on $k$-mer frequency counts. Technical proofs appear
in Appendix A.  Our main result, Theorem \ref{thm:gentrace}, is quite
general, applying to arbitrary pairwise distributions that are at
stationarity.  We use this result to derive corrected distance
calculations for the Juke-Cantor model and the Kimura 2- and
3-parameter models.  These corrections yield statistically consistent
estimates of evolutionary times between extant taxa.  Coupled with a
statistically consistent method for constructing a tree from distances
(for example, Neighbor Joining \citep{Saitou1987}), this produces a
statistically consistent method for reconstructing phylogenetic trees
from $k$-mer counts.

Let $S$ be a sequence on an $L$-letter alphabet, $[L] := \{1,2,
\ldots, L\}$.  For a natural number $k$, let $X$ denote the vector of
$k$-mer counts extracted from $S$.  That is, for each $W =
w_{1}w_{2}\ldots w_{k} \in [L]^{k}$ the coordinate $X^{W}$ records the
number of times that $W$ occurs as a contiguous substring in $S$.  A
standard $k$-mer method computes a distance between two sequences
$S_{1}$ and $S_{2}$ of lengths $n_{1}$ and $n_{2}$ by first computing
their respective $k$-mer vectors $X_{1}$ and $X_{2}$ and then
computing the squared-Euclidean distance
$$
||X_1-X_2||_2^2=\sum_{W \in [L]^{k}} \left(X^{W}_{1} - 
X^{W}_{2} \right)^{2}.
$$

Consider two sequences descended from a common ancestor while
undergoing a base-substitution process described by standard
phylogenetic modeling assumptions. More specifically, we may assume
one of the sequences, $S_1$, is ancestral to the other, $S_2$, and its
sites are assigned states in $[L]$ according to an i.i.d.~process with
state probability vector $\pi=(\pi^w)_{w\in L}$. Additionally, $\pi$
is the stationary distribution of an $L\times L$ Markov matrix $M$
describing the single-site state change process from sequence $S_1$ to
sequence $S_2$. For continuous-time models, with rate matrix $Q$ and
time (or branch length) $t$, one has $M = \exp(Qt)$.  The probability
of a $k$-mer $W = w_{1}w_{2}\ldots w_{k} \in [L]^{k}$ in any $k$
consecutive sites of either single sequence is then $\pi^{W}
=\prod_{j=1}^k \pi^{w_j}$. The $k$-mer vectors $X_{1}$ and $X_{2}$ are
random variables which summarize $S_1$ and $S_2$.

The following theorem relates the expectation of an appropriately
chosen norm of the difference of $k$-mer counts $X_1-X_2$ to the
base-substitution model.  Since the expectation can be estimated from
$k$-mer data, this means that from $k$-mer data we can infer
information on how much substitution has occurred.

\begin{thm} \label{thm:gentrace}
Let $S_{1}$ and $S_{2}$ be two sequences of length $n$
 generated from an indel-free Markov model
with transition matrix $M$ and stationary distribution $\pi$,
and let $X_{1}$ and $X_{2}$ be the resulting $k$-mer count vectors.
Then
\begin{equation}
\mathbb{E}\left[  \sum_{W \in [L]^{k}}  \frac1{\pi^W} (X_{1}^{W} - X_{2}^{W}) ^{2} \right]  = 2(n-k+1)(L^k-(\tr M)^k ).\label{eq:gentrace}
\end{equation}
\end{thm}

Since for each $W$ the random variable $X_1^W-X_2^W$ has mean 0, the
expectation on the left of equation \eqref{eq:gentrace} can be viewed
as a (weighted) variance of the $k$-mer count difference.  Indeed this
observation plays an important role in the proof, which appears in
Appendix A.

We now derive consequences for the Jukes-Cantor model.
In this setting the rate matrix $Q$ has the form:
$$
Q = \begin{pmatrix}
-3\alpha & \alpha & \alpha & \alpha \\
\alpha & -3\alpha &  \alpha & \alpha \\
\alpha &   \alpha &-3\alpha &  \alpha \\
\alpha  &  \alpha & \alpha & -3\alpha
\end{pmatrix}.
$$
In the Jukes-Cantor model, the rate parameter $\alpha$ and the branch length 
$t$ are confounded
with only their product $\alpha t$ identifiable.
For simplicity we set $\alpha = 1/3$ which gives the branch length
$t$ the interpretation of the expected number of substitutions per site.  
The stationary distribution
is  uniform.   Theorem \ref{thm:gentrace} then implies the following.

\begin{cor}\label{cor:jccorrection}
  Let $S_{1}$ and $S_{2}$ be sequences of length $n$ generated under
  the Jukes-Cantor model on an edge of length $t$. Let $X_{i}$ be the
  $k$-mer count vector of $S_{i}$ and let $d = \mathbb{E}\left [ \|
    X_{1} - X_{2} \|_{2}^{2}\right]$ be the expected squared Euclidean
  distance between the $k$-mer counts.  Then
\begin{equation}
t = - \frac{3}{4}   \ln \left( \frac{4}{3}  \sqrt[k]{1 - \frac{d}{2(n-k+1)} }  - \frac{1}{3}  \right).\label{eq:JCdist}
\end{equation}
\end{cor}  

Equation \eqref{eq:JCdist} thus gives a model-corrected estimate of
the branch length $t$ under the Jukes-Cantor model, when in place of
the true expected value $d$ one uses an estimate obtained from data.
\begin{proof}[Proof of Corollary \ref{cor:jccorrection}]
To specialize Theorem \ref{thm:gentrace}
to the Jukes-Cantor model, take $L = 4$,  and 
$\pi^W =  4^{-k}$ for all $W \in \{{\tt A},{\tt C},{\tt G},{\tt T}\}^{k}$.
Dividing both sides of equation \eqref{eq:gentrace} by $4^{k}$ we deduce that
\begin{equation}
d  =  2(n-k+1) \left(1 - \left({\tr M}/{4}\right)^{k} \right).\label{eq:jcderiv}
\end{equation}
For the Jukes-Cantor model 
$$
M  =  \exp(Qt) =  \begin{pmatrix} y&x&x&x\\x&y&x&x\\x&x&y&x\\x&x&x&y\end{pmatrix}
$$
with
\begin{equation}  \label{eq:xy}
x= \frac{1 -   \exp(-4t/3)}{4},\ \ y=  \frac{1 + 3 \exp(-4t/3)}{4}, 
\end{equation}
so that ${\tr M}  =  {1 + 3  \exp(-4t/3)}$. 
Substituting this into equation \eqref{eq:jcderiv} 
and solving for $t$ yields the desired formula.
\end{proof}

Next we derive an analogous result for the Kimura $3$-parameter model, with rate 
matrix
$$
Q = \begin{pmatrix}
*  &  \alpha & \beta & \gamma  \\
\alpha &  *  & \gamma & \beta  \\
\beta & \gamma & *  & \alpha  \\
\gamma & \beta & \alpha &  *
\end{pmatrix}.
$$

\begin{cor}
Let $S_{1}$ and $S_{2}$ be two random sequences of length $n$ generated
under the Kimura 3-parameter model on an edge of length $t$.   
Let  $X_{i}$ be the $k$-mer count vector of $S_{i}$.
Then 
$$
\mathbb{E}[ \| X_{1} - X_{2} \|_{2}^{2}]  = 2(n - k + 1)
\left(1 - \left( \frac{ 1 + e^{-2(\alpha + \beta)t} + e^{-2(\alpha + \gamma)t}  
+ e^{-2(\beta + \gamma)t}}{4}\right)^{k}  \right).
$$
\end{cor}
Note that the right side of this equation is strictly increasing
as a function of $t$. Thus if $\alpha$, $\beta$, $\gamma$ 
are known, and $\mathbb{E}[ \| X_{1} - X_{2} \|_{2}^{2}]$  is estimated, 
it is straightforward to estimate $t$ using 
a numerical root finding
algorithm.

For general rate matrices $Q$, the matrix $M = \exp(Qt)$ has trace
\begin{equation}\label{eq:eigformula}
\tr M =  \sum_{i = 1}^L  e^{\lambda_i t}
\end{equation}
where $\lambda_1, \ldots, \lambda_L$ are the eigenvalues of $Q$,
counted with multiplicity.  Since $Q$ is a rate matrix, all these
eigenvalues have nonpositive real part.  If all the eigenvalues are
real, then equation \eqref{eq:eigformula} shows $\tr M$ is a
decreasing function of $t$.  This means we can consistently estimate
the branch length if we assume $Q$ is known and we have an estimate
for the expectation in equation \eqref{eq:gentrace}.  For instance,
this argument shows that for any time-reversible rate matrix (i.e.,
from the general time-reversible model GTR) we can obtain
statistically consistent estimates for the branch lengths.

%%%%%%%%%%%%%%%%%%%%%%%%%%%%%%%%%%
%%%%%%%%%%%%%%%%%%%%%%%%%%%%%%%%%%
%%%%%%%%%%%%%%%%%%%%%%%%%%%%%%%%%%
%%%%%%%%%%%%%%%%%%%%%%%%%%%%%%%%%%

\section{Jukes-Cantor Correction}\label{sec:JCcor}

In this section, we give a detailed explanation of the statistical 
consistency for phylogenetic tree reconstruction using our
Jukes-Cantor correction from Corollary \ref{cor:jccorrection}.
In particular, we explain that without this correction, even with
arbitrary amounts of data generated from the model, the $k$-mer
method based on the squared Euclidean distance is statistically
inconsistent for every $k$.

Corollary \ref{cor:jccorrection} gives an estimate of branch lengths
under the Jukes-Cantor model based on the value of $d =
\mathbb{E}\left [ \| X_{1} - X_{2} \|_{2}^{2}\right]$. Applying the
same formula to an empirical estimate $\hat d$ of $d$, it can thus be
viewed as giving a model-based distance correction to the naive
distance estimate $\hat d$. This is similar to the usual Jukes-Cantor
correction applied to the frequency $\hat p$ of mismatches of bases in
aligned sequences.  When $k = 1$, equation \eqref{eq:JCdist}
simplifies to
$$
t =  - \frac{3}{4}   \ln \left(  1 - \frac{4}{3} \cdot \frac{d}{2n}  \right)
$$
which is clearly very similar to the usual Jukes-Cantor correction obtained from an
alignment with $\frac{d}{2n}$ playing the role  of $p$.
  
That $\frac{d}{2} = p$ for $k = n =1$ can be justified rigorously as
follows: For a single aligned site in two sequences, the probability
of a mismatch is $p$ under the Jukes-Cantor model.  The $k$-mer count
vectors $X_1$ and $X_2$ are the elementary basis vectors $X_1 =
\mathbf e_i$ and $X_2 = \mathbf e_j$, and the quantity $\| X_1 - X_2
\|_2^2$ is $0$ or $2$ depending if $i = j$ or $i \neq j$.  Thus, the
expected value $d = \mathbb{E}\left [ \| X_{1} - X_{2}
  \|_{2}^{2}\right] = (1-p) \cdot 0 + p \cdot 2 = 2p$.  It follows
that our estimate for the branch length $t$ is exactly the
Jukes-Cantor corrected estimate when $1$-mer frequencies at each site
are used to estimate $d$.  Indeed, formula \eqref{eq:JCdist} gives a
natural generalization of the pairwise corrected distance to the
present context of $k$-mers.

To understand the potential impact of the correction of Corollary
\ref{cor:jccorrection} we first work theoretically, by assuming we
have the true expected value $d$ in hand.  Later, in Section
\ref{sec:simulation}, we use simulations to investigate the usefulness
of the branch length estimate \eqref{eq:JCdist} with finite length
sequences, to understand its practical impact.

We follow the framework suggested by \citet{Felsenstein1978}.
We consider an unrooted four-leaf tree with topology $12|34$.  
Two branch lengths $t_a$ and $t_b$, each ranging over the interval $(0, \infty)$,
are used, with $t_a$ on edges $2|134$, $3|124$, and $12|34$ and
$t_b$ on the edges $1|234$ and $4|123$.  This tree is depicted
in Figure \ref{fg:felsen}.  The branch lengths are transformed to 
probabilities $a$ and $b$ in $(0, .75)$, probabilities that bases at a 
site differ at opposite ends of a branch.

Consider the naive $k$-mer method that uses $d$ as a distance together
with the 4-point condition (or equivalently, Neighbor Joining) to
infer a tree topology. To analyze its behavior, we must first relate
the expected values of
$$
d  =  2(n-k+1)\left (1 - \left(\frac{1 + 3  \exp(-4t/3)}{4}\right )^{k} \right),
$$
for each taxon pair to the underlying branch parameters $a$ and $b$.
As $a$ is
the probability that some change from the current state is made along the
edge of scaled length $t_a$ ($y$ from equation \eqref{eq:xy}), we have that 
the diagonal element from the associated Jukes-Cantor transition matrix is
\begin{equation}\label{eq:transformed}
1-a  =  \frac{1 + 3  \exp(-4t_{a}/3)}{4}
\end{equation}
and thus
$$
t_{a}  =  - \frac{3}{4} \ln \left (1-\frac{4}3 a \right ).
$$
Similar arithmetic gives the formula for $t_b$.

This yields the following formulas for the expected distance in
terms of the parameters $a, b$:
\begin{eqnarray*}
d_{12} = d_{34}  & = & {\textstyle 2(n-k+1) \left(1 - \left(\frac{1 + 3(1-\frac{4}{3}a) (1-\frac{4}{3}b) }{4}\right)^{k} \right)},\\
d_{13} = d_{24}  & = & {\textstyle 2(n-k+1) \left(1 - \left(\frac{1 + 3(1-\frac{4}{3}a)^{2} (1-\frac{4}{3}b) }{4}\right)^{k} \right) },\\
d_{14}     & = &  {\textstyle 2(n-k+1) \left(1 - \left(\frac{1 + 3(1-\frac{4}{3}a) (1-\frac{4}{3}b)^{2} }{4}\right)^{k} \right) },\\
d_{23}     & = &  {\textstyle 2(n-k+1) \left(1 - \left(\frac{1 + 3(1-\frac{4}{3}a)^3 }{4}\right)^{k} \right). }
\end{eqnarray*}

To construct correctly the unique
true tree $12|34$ using the $4$-point condition or Neighbor Joining
requires that these distances satisfy
$$
d_{12} + d_{34} <  \min( d_{13} + d_{24}, d_{14} + d_{23} ).
$$
Note that $d_{12} + d_{34} < d_{13} + d_{24}$ for all $a>0$,
so we focus on the condition
$$d_{12} + d_{34} < d_{14} + d_{23}.$$  
Using the formulas above, this becomes:
$$
2\left(1 + 3(1-{\textstyle \frac{4}{3}} a) (1-{\textstyle \frac{4}{3}} b) \right)^{k}  \geq
 \left(1 + 3(1-{\textstyle \frac{4}{3}} a) (1-{\textstyle \frac{4}{3}} b)^{2}  \right)^{k} 
+ 
\left(1 +  3(1-{\textstyle \frac{4}{3}} a)^3 \right)^{k}. 
$$
The values of $a$, $b$ for which this is satisfied are shown by the
white regions in Figure \ref{fg:felsen}, for $k = 1,3,5$.  As $k$
increases the white regions change; when $k=1$ the boundary curve is a
circle, and as $k\to\infty$ it approaches a parabola with vertex in
the upper right corner, passing through the lower left.  Note that the
white region indicates where the naive $k$-mer distance inference
behaves well provided one knows $d$ exactly --- in practice one only
has an estimate of $d$ and should not expect even this good behavior.

In contrast, using the corrected Jukes-Cantor $k$-mer distance from
equation \eqref{eq:JCdist} to make diagrams analogous to those of
Figure \ref{fg:felsen} would show the entire square white.  If $d$
were known exactly, inference would be perfect.  The corrected
distances lead to statistically consistent distance methods on
$4$-taxon trees.  More generally, our argument in Section
\ref{sec:formulas} shows that we can use Theorem \ref{thm:gentrace} to
derive statistically consistent estimates for the evolutionary time
between species when we have a known time-reversible rate matrix $Q$.

\begin{figure}[h]
\resizebox{!}{3cm}{\includegraphics{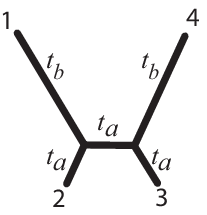}}
\resizebox{!}{3cm}{\includegraphics{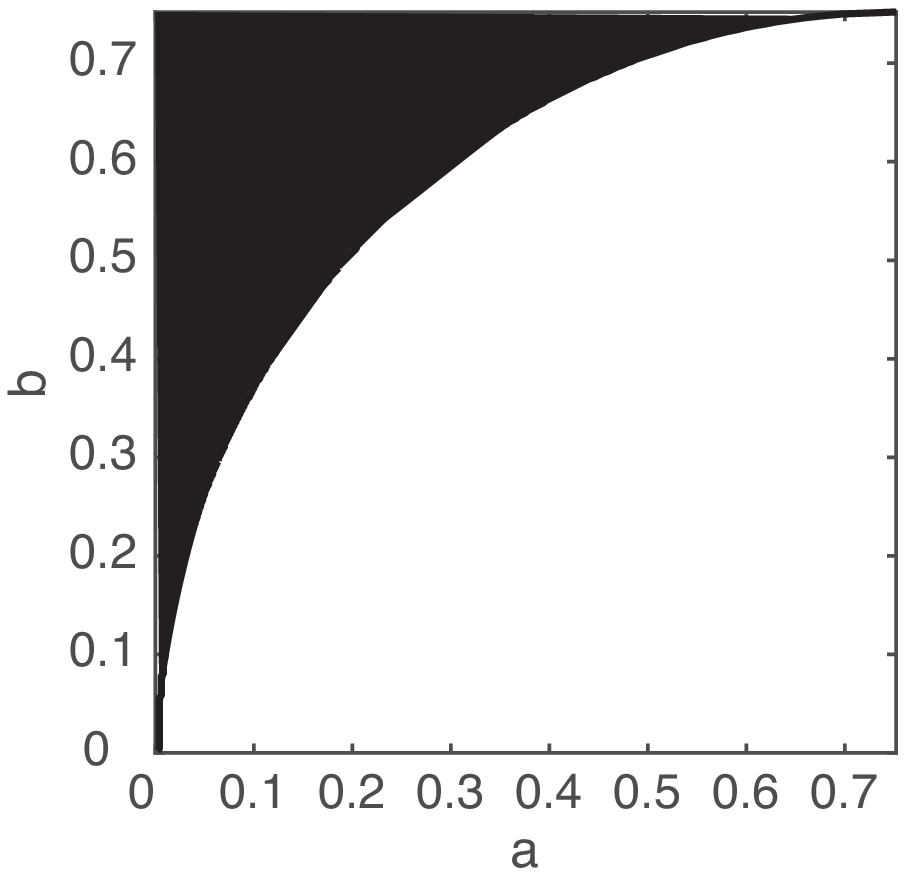}} 
\resizebox{!}{3cm}{\includegraphics{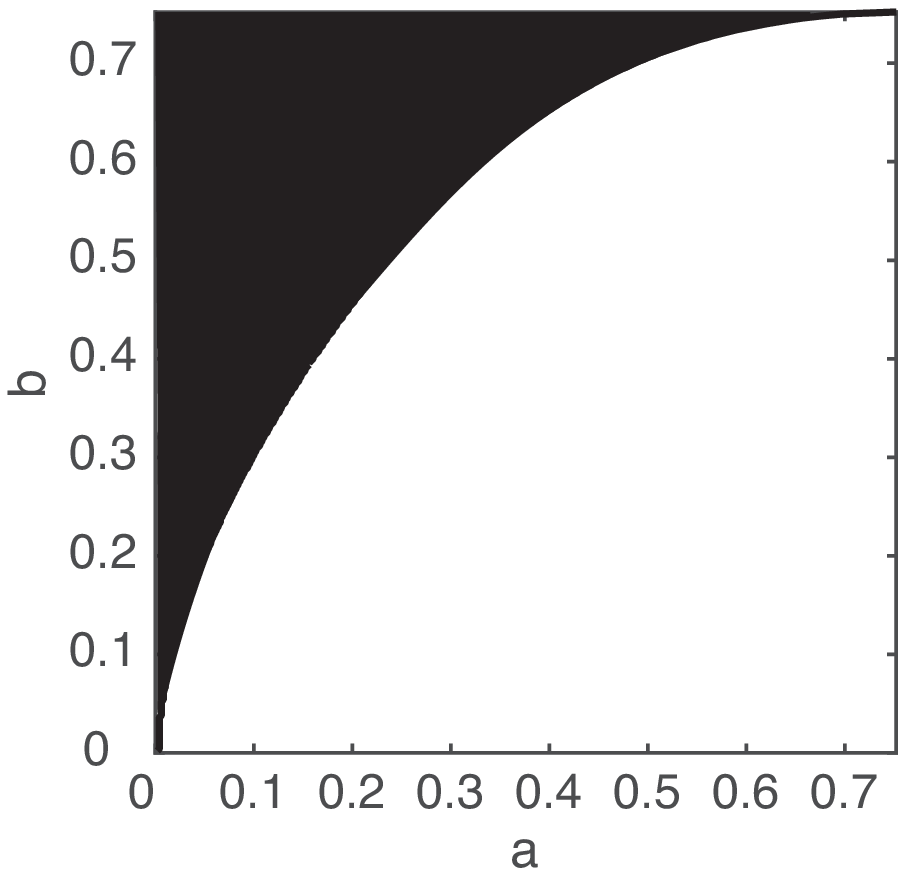}}
\resizebox{!}{3cm}{\includegraphics{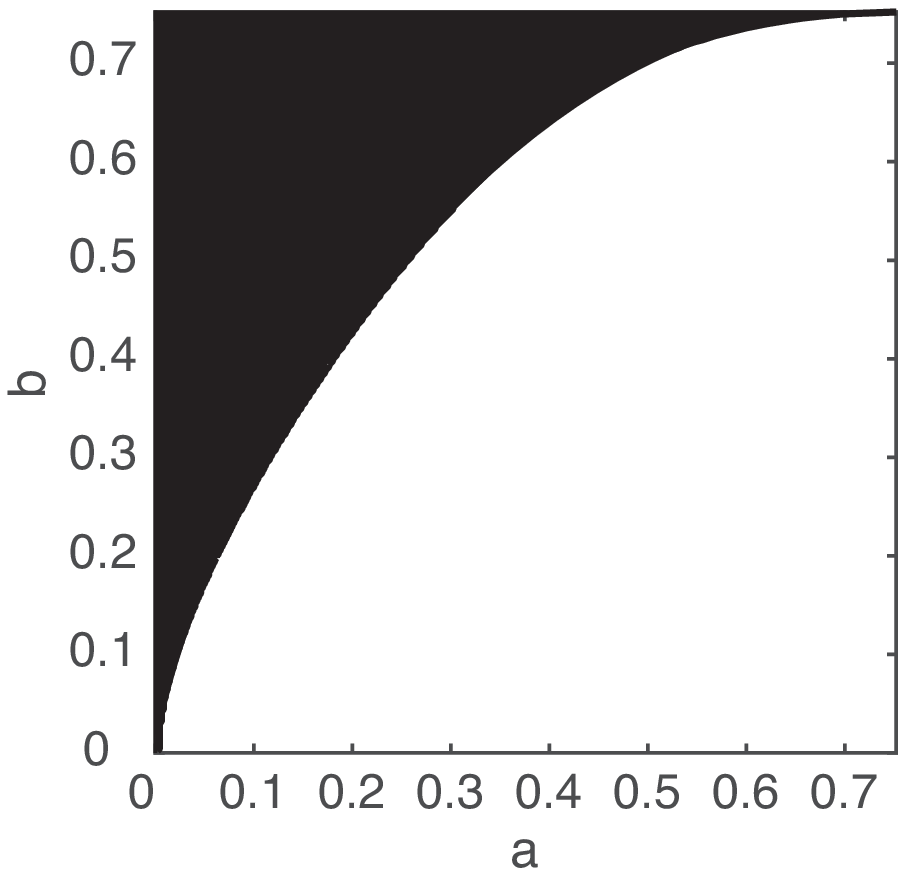}} 
\caption{White region is the zone of consistency for tree inference using the naive $k$-mer distance combined with the 4-point condition. From left to right, $k =1, 3$, and $5$. The usual ``Felsenstein Zone" is in the upper left.}\label{fg:felsen}
\end{figure}

%%%%%%%%%%%%%%%%%%%%%%%%%%
%%%%%%%%%%%%%%%%%%%%%%%%%%%
%%%%%%%%%%%%%%%%%%%%%%%%%%
%%%%%%%%%%%%%%%%%%%%%%%%%%%
%%%%%%%%%%%%%%%%%%%%%%%%%%
%%%%%%%%%%%%%%%%%%%%%%%%%%%

\section{Identifiability of indel-free model parameters}\label{sec:ident}

The results in Section \ref{sec:formulas} prove that, with knowledge of the stationary 
base frequency $\pi$ of an unknown Markov matrix $M$ describing base
substitutions from one sequence to another, $\tr M$ is identifiable from the joint 
distribution of $k$-mer counts in the two sequences. If one assumes a continuous 
time model with $M= \exp(Qt)$ and $Q$ a known time-reversible
rate matrix, then this is sufficient to identify 
lengths $t$ between taxa on the tree.  As a consequence, with $Q$ known 
the metric phylogenetic tree relating many taxa is identifiable.

In fact, more is true: $\pi$ and $M$ are identifiable from $1$-mer count 
distributions as well.  This is the result in the next
proposition, which in addition to being interesting in its own right, 
plays a role in the proof of Theorem \ref{thm:gentrace}. 
Its proof appears in Appendix B.
Note that in this result we do not assume that base frequency
distributions $\pi_i$ 
are the stationary vectors of the Markov matrix.

\begin{prop} \label{prop:joint} From the joint distribution of 1-mer
  count vectors $X_1$ and $X_2$ of two sequences $S_1$ and $S_2$ of
  length $n$, one can identify the distributions $ \pi_1$ and $\pi_2$
  of bases in each sequence, and the joint distribution
  $P=\diag(\pi_1)M$ of bases at a single site in the two sequences.
  Specifically, $\pi_i =\frac 1n \EE[X_i ]$, and for $w,u\in [L]$,
$$P_{wu}= \frac 12\left ( \pi_1^w+\pi_2^u-\frac {\pi_1^w\pi_2^u}{n}\EE\left [ \left (\frac {X_1^w} {\pi_1^w}-\frac {X_2^u  }{\pi_2^u}\right )^2\right ] \right).$$
\end{prop}

The formula for $P_{wu}$ in this proposition ultimately underlies our
suggested practical inference method. However, there is a simpler
formula, applying for any $k$, showing that from a joint $k$-mer count
vector distribution one can identify the joint probabilities $P_{wu}$:
For sequences of length $n$ and the particular $k$-mers $W=www\ldots
w$ and $U=uuu\ldots u$,
$$
\Prob(X_1^W=n-k+1, \, X_2^U=n-k+1)= {(P_{wu})}^n.$$ 
Of course the method of estimation suggested by this approach
is useless in practice, since it is based on events that are rarely, if ever, observed.

Nonetheless, since $P$ and $\pi_1$ can be found from the joint
distribution of $X_1$ and $X_2$ for any $k$, the transition matrix
$M=\diag(\pi_1)^{-1}P$ is also identifiable.  In the continuous-time
model setting, where $M=\exp (Qt)$, $Q$ can be found, first up to a
scalar multiple, and then normalized.  Putting this together yields
the following.

\begin{thm} For an indel-free GTR model, all parameters, both
  numerical ones and tree topology, can be identified from pairwise
  joint $k$-mer count distributions.
\end{thm}

If we consider sequences three-at-a-time, rather than pairwise, we obtain an
analog of Proposition \ref{prop:joint}, again without assuming stationarity.
This new result is based on third moments, rather than second, and its proof  is given in Appendix B.

\begin{prop}\label{prop:joint3}
  For a 3-leaf tree, the joint distribution $P=(P_{uvw})$ of site
  patterns is identifiable from the joint 1-mer count vector
  distributions of the 3 taxa.  Specifically, define a random
  variable $$Y_{uvw}=\alpha {X_1^u} +\beta X_2^v +\gamma X_3^w,$$
  where $\alpha, \beta,\gamma$ are constants chosen so
$$ \alpha{\pi_1^u}+ \beta{\pi_2^v}+  \gamma{\pi_3^w}=0.$$
Then
\begin{multline}\label{eq:3way}
P_{uvw}= \frac 1{6\alpha\beta\gamma n}\EE(Y_{uvw}^3)+\frac 12 \left (\frac {\alpha+\beta}\gamma P_{uv+} +  \frac {\alpha+\gamma}\beta P_{u+w}+\frac {\beta+\gamma}\alpha P_{+vw}\right )\\
-\frac 16\left ( \frac {\alpha^2}{\beta\gamma}\pi_1^u+\frac {\beta^2}{\alpha\gamma}\pi_2^v+ \frac {\gamma^2}{\alpha\beta}\pi_3^w\right ) .
\end{multline}
where the pairwise marginal distributions $P_{uv+}$, $P_{u+w}$, and $P_{+vw}$ in equation
 \eqref{eq:3way} are identifiable by Proposition \ref{prop:joint}. 
\end{prop}
Proposition \ref{prop:joint3} is significant in that it establishes
that the distribution of $1$-mer counts contains enough information to
identify parameters of more general models than our preceding
arguments allow.  Recall, for instance, that parameters for the
General Markov (GM) model, in which the base substitution process on
each edge of the tree can be specified by a different Markov matrix,
are identifiable from the marginalization of the site pattern
distribution to $3$-taxon sets \citep{Chang96}, but are not
identifiable from pairwise marginalizations.  In the present context
of $k$-mers, we obtain the following.

\begin{cor} For an indel-free GM model,
all parameters, both numerical ones and tree topology, are identifiable from the joint $1$-mer count vector distributions on $n$ taxa. 
\end{cor}

%%%%%%%%%%%%%%%%%%%%%%%%%%%%%%%%%%%%%%%%%%%%%%%%
%%%%%%%%%%%%%%%%%%%%%%%%%%%%%%%%%%%%%%%%%%%%%%%%
%%%%%%%%%%%%%%%%%%%%%%%%%%%%%%%%%%%%%%%%%%%%%%%%

\section{Practical $k$-mer distances between sequences}\label{sec:practical}

In this section, we apply the results of Section \ref{sec:formulas} to
develop practical methods for estimating pairwise distances between
sequences.  Those derivations were made under the assumption that
sequences evolved in the absence of an indel process, and thus that
sequences could be unambiguously aligned.  In practice, however, we
desire a method of distance estimation that can be applied in the
presence of a mild indel process, without a precise alignment.
Although this violates our model assumptions, in Section
\ref{sec:simulation} we use simulations to investigate how robust our
resulting method is to such a violation.

\smallskip

Assuming a Jukes-Cantor process of site substitution and no indel
process, formula \eqref{eq:JCdist} of Corollary \ref{cor:jccorrection}
suggests a natural definition for a distance, provided we have a good
method of approximating $d = \mathbb{E}\left [ \| X_{1} - X_{2}
  \|_{2}^{2}\right]$.  If the observed values of the random variables
$X_1$ and $X_2$ are denoted in lower case, so $x_{1}$ and $x_{2}$ are
observed $k$-mer count vectors, then one could simply compute
\begin{equation}\label{eq:badestimate}
 \| x_{1} - x_{2} \|_{2}^{2}=\sum_{W \in [L]^{k}}  \left (x^{W}_{1} - x^{W}_{2}\right )^{2}
\end{equation}
as a point estimate for $d$.  This is a very poor estimate for the
expected value, however, since only one sample ($\|x_1 - x_2\|^2$) is
used to estimate a mean.  Indeed, this estimate has large
variance. Moreover, naively increasing sequence length (number of
$k$-mers) would do nothing to address the fundamental problem of
needing more samples to estimate a mean well.

To obtain a better estimate of $d$, with smaller variance, we instead
subdivide the two sequences into a fixed number $B$ of contiguous
blocks.  Assuming for $1 \le i \le B$ that the $i$th blocks of the two
sequences are at least roughly orthologous, we compute the $k$-mer
frequencies $x_{j,i}$ for each block $i$ in sequence $j$.  Then the
values of $\| x_{1,i} - x_{2,i} \|_{2}^{2}$ for the $B$ blocks can be
averaged to estimate $d$.  In this framework, we are adopting the
approach introduced by \citet{Daskalakis2013}.

We have in mind two scenarios for using this approach on data, which
are displayed in Figure \ref{figure:scenarios}.  The first is under
the assumption that if indels occurred, they were distributed evenly
over the sequences. Then if the blocks are defined as a fixed fraction
of the full sequence lengths, most of the sites in the $i$th blocks of
the two sequences will be orthologous.  The second is that the blocks
arise naturally in the data; for instance if a dataset consists of
multiple genes, then each gene can be treated as a block.  In this
case, the point estimates for each gene would be averaged over all
genes, making appropriate adjustments for their varying lengths.

\smallskip

\begin{figure}[h]
\begin{center}
\includegraphics[width=7cm]{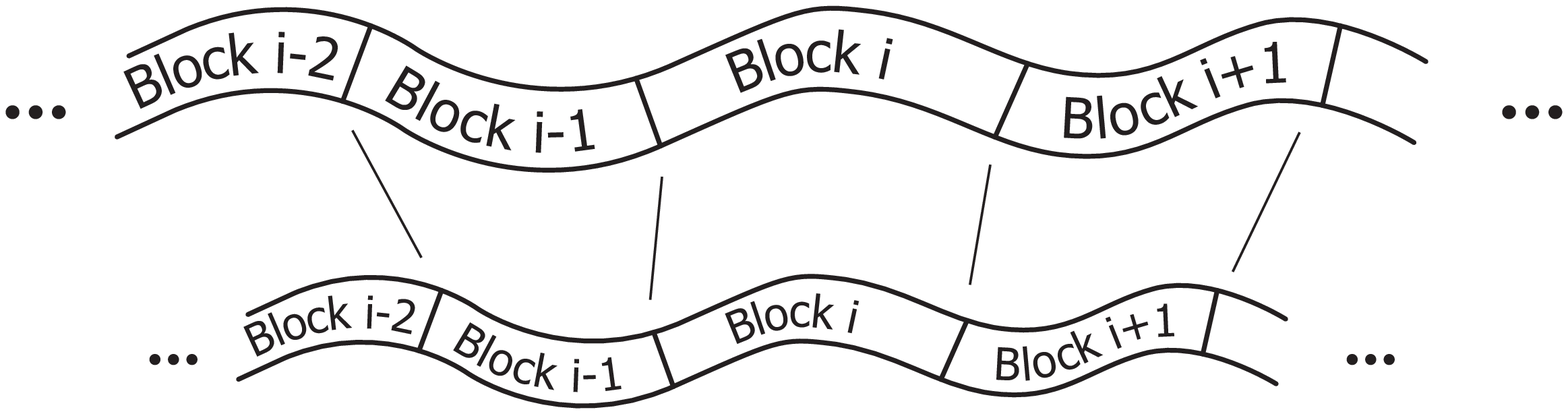}
\hskip 1cm
\includegraphics[width=7cm]{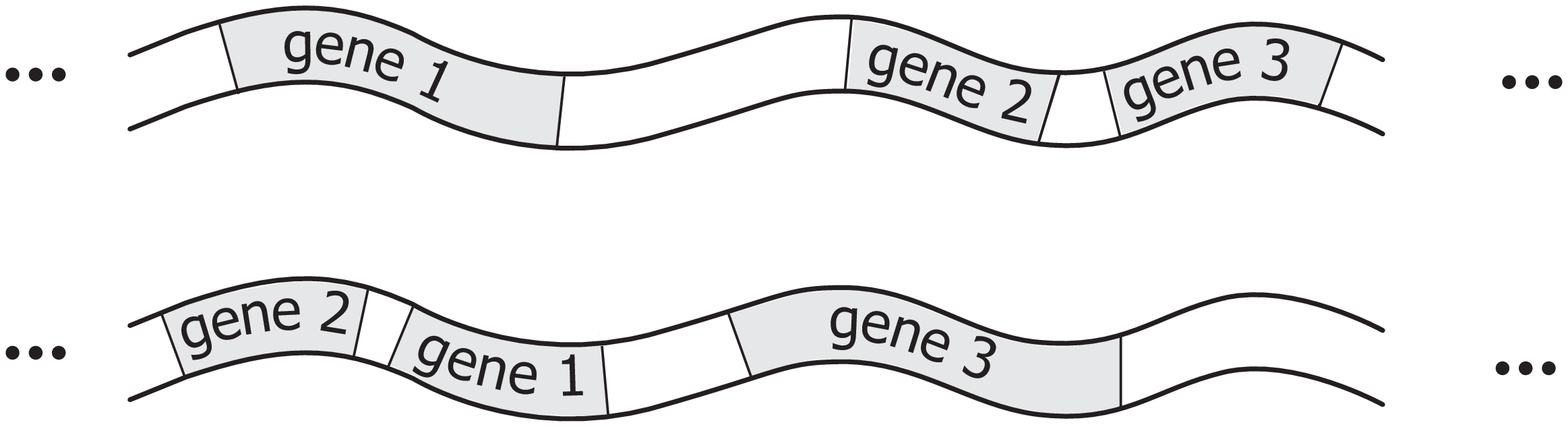}
\caption{\label{figure:scenarios} On the left, two sequences in which blocks $i$ are
roughly orthologous, perhaps due to a uniform indel process.  
On the right, two genomes in which 
genes serve as blocks for data analysis.}
\end{center}
\end{figure}

To be precise, in addition to specifying $k$, under the first scenario
we must also specify a number $B$ of blocks to be used in our
calculations. To subdivide a sequence $S_j$ of length $n_j$ as
uniformly as possible, each block will have length $n_{j,i}=n_j/B$,
suitably rounded for $1\le i \le B$, so block lengths for a single
sequence can differ at most by one.  Under the second scenario, using
natural blocks like genes, the length $n_{j,i}$ is specified by the
data, and will vary more widely.

Now for block $i$ in sequence $j$, let $x_{j,i}$ be the  $k$-mer count 
vector and $\mu_{j,i}=
(n_{j,i}-k+1)/4^k$ the mean $k$-mer count under the Jukes-Cantor model. We define
\begin{equation}
\tilde d  = \frac 1B\sum_{i=1}^B  \bigg\| \frac{(x_{1,i}-\mu_{1,i})}{\sqrt{n_{1,i}-k+1}} - 
\frac{(x_{2,i}-\mu_{2,i})}{\sqrt{n_{2,i}-k+1}} \bigg\|_{2}^{2}.
\label{eq:destimate}\end{equation}
Note that in this formula both the centering of $x_{j,i}$ by
subtracting $\mu_{j,i}$ and the normalization by dividing by the
square root of the number of $k$-mers depend upon the length
$n_{j,i}$.  In the special situation where $n_{j,i}=n$ for all $i,j$ ,
and hence $\mu_{j,i}=\mu$, this reduces to
$$
\tilde d  =\left (\frac1 {{n-k+1}} \right ) \frac 1B\sum_{i=1}^B  {\| x_{1,i} - x_{2,i} \|_{2}^{2}}   \approx \frac d {{n-k+1}} .$$
Comparing this estimate for $\tilde d$ to equation \eqref{eq:JCdist}, it is natural
to define a Jukes-Cantor $k$-mer distance $d_{JC}^{k,B}$, dependent on $k$ and $B$, by
\begin{equation}
d_{JC}^{k,B}= - \frac{3}{4}   \ln \left( \frac{4}{3}  
\sqrt[k]{1 - \frac{\tilde d}{2} }  - \frac{1}{3}  \right).\label{eq:JCkBdist}
\end{equation}
We use this formula extensively in the simulations whose results are presented in the next section.

In examining \eqref{eq:JCkBdist}, it is unclear \emph{a priori} which
values of $k$ and $B$ will yield the best estimate for $d_{JC}^{k,B}$.
In the particular case that sequences evolved without an indel
process, the lowest variance estimate of $d_{JC}^{k,B}$ is obtained by
taking the largest number $B$ of samples, i.e.~ each block has length
$k$ (the smallest possible length which allows $k$-mers to be
counted).  However, in the presence of an evolutionary indel process a
true alignment of sequences would contain gaps, and such short block
sizes would give poor results. For good performance, we need the $i$th
blocks in the two sequences to be composed mostly of orthologous
sites.  If the block size is small, this is unlikely to be true, as
even a mild indel process might result in orthologs residing in
different blocks.  The art is to find the right compromise between a
large number $B$ of blocks and a large enough length $n_{j,i}$ for
each block to ensure many orthologs.  Results of simulation studies in
the next section confirm this trade-off.

Using $1$-mer distributions and taking into account a particular model
of the indel process, \cite{Daskalakis2013} give a detailed analysis
of a distance method along the lines described here.  Their results
suggest that the block sizes should be of size roughly the square root
of total sequence length.  While the approach of Daskalakis and Roch
inspired our results, since our approach to a $k$-mer distance is
based on a model without indels, and our extension to a distance
formula for sequence evolution in the presence of indels is heuristic,
we can offer no such guidance.  A fruitful direction for future
research is to explore $k$-mer distances under some explicit model of
sequence evolution with indels.

%%%%%%%%%%%%%%%%%%%%%%%%%%%%%%%%%%
%%%%%%%%%%%%%%%%%%%%%%%%%%%%%%%%%%
%%%%%%%%%%%%%%%%%%%%%%%%%%%%%%%%%%
%%%%%%%%%%%%%%%%%%%%%%%%%%%%%%%%%%

\section{Simulation studies}\label{sec:simulation}

\subsection*{Methods}

We performed extensive simulations to attempt to understand how the distance formula
in \eqref{eq:JCkBdist} might work in practice and to compare distance
methods with $d_{JC}^{k,B}$ to other alignment-free methods
for reconstructing phylogenetic trees from sequence data. Data was simulated using
the sequence evolution simulator INDELible \citep{Fletcher2009},
which produces sequence data under standard 
base substitution models with or without an additional insertion and deletion process. 

All of our simulations use the Jukes-Cantor (JC) substitution model on
$4$-taxon tree.  We consider only trees in which two branch lengths
occur, $t_a$ and $t_b$, as shown in Figure \ref{figure:abtree}.  This
allows us to investigate performance over an important range of
parameter space, yet still display the success of an algorithm in an
easily-interpretable $2$-dimensional display, as introduced by
\citet{Huelsenbeck1995}.

The two branch lengths $t_a$ and $t_b$ each range over the interval
$(0, \infty)$, but are transformed to probabilities $a$ and $b$ in
range $(0, .75)$, probabilities that bases at a site differ at
opposite ends of a branch (see equation \ref{eq:transformed}).  In
this interval, we sampled points from $.01$ to $.73$, with increments
of $.02$, to get a $37 \times 37$ grid of transformed branch lengths.
For each choice of branch lengths we generated 100 sets of four
sequences, used a specific method to recover the tree topology, and
recorded the frequency the method under study reconstructs the correct
tree $12 \vert 34$ from the simulated data.

\begin{figure}[h] 
\begin{center}
\ifthenelse{\boolean{colour}}{

\resizebox{!}{3cm}{\includegraphics{abtree.eps}}
\quad 
\resizebox{!}{3cm}{\includegraphics[height=3cm]{figure_6_1_1_a18_M100_idr_05_kmer_3.eps}} 
\quad 
\resizebox{!}{3cm}{\includegraphics[height=3cm]{figure_6_1_2_a18_M100_idr_05_nj.eps}}
\, 
\resizebox{!}{!}{\includegraphics[height=3cm]{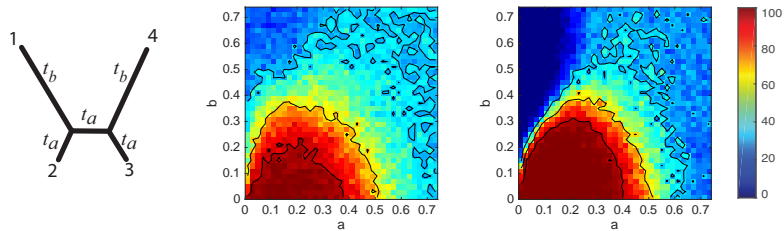}}
}{
\resizebox{!}{3cm}{\includegraphics{abtree.eps}}
\quad 
\resizebox{!}{3cm}{\includegraphics[height=3cm]{figure_6_1_1_a18_M100_idr_05_kmer_3_BW.eps}} 
\quad 
\resizebox{!}{3cm}{\includegraphics[height=3cm]{figure_6_1_2_a18_M100_idr_05_nj_BW.eps}}
\, 
\resizebox{!}{!}{\includegraphics[height=3cm]{color_bar_BW.eps}}
}
\caption{\label{figure:abtree}Figure on the left displays the model tree used for simulations.
The middle and right figures are representative Huelsenbeck diagrams for some (unspecified)
methods of inference. 
The horizontal axis is labeled by $a$ and the vertical one with $b$, both in the range $(0,.75)$
after transformation.  Contour lines are drawn at levels $.95$, $.67$, and $.33$.
The figure on the right suggests significant long branch attraction is present, as witnessed 
by the strong bias against the correct tree (much less than 33\% correct) along the upper left side. }
\end{center}
\end{figure}

The middle and the right plots in Figure \ref{figure:abtree} show
typical Huelsenbeck diagrams presenting results from such simulations.
The \ifthenelse{\boolean{colour}}{dark red}{white} regions correspond
to regions where the method reconstructs the true tree topology, with
split $12|34$, close to $100\%$ of the time.
\ifthenelse{\boolean{colour}}{Dark blue}{Black} regions are regions
where inference is strongly biased against the correct tree,
reconstructing it close to $0\%$ of the time.
\ifthenelse{\boolean{colour}}{Light blue}{Dark grey} corresponds to a
method constructing the true tree correctly about $33\%$ of the time;
that is,~the method is indistinguishable from the process of randomly
picking the tree topology to return.  For any phylogenetic method
applied to simulated sequence data, one typically sees
\ifthenelse{\boolean{colour}}{light blue}{dark grey} in the upper
right of these figures ($a\approx b\gg 0$),
\ifthenelse{\boolean{colour}}{darker blue}{darker grey to black} in
the upper left ($b\gg a$) in the ``long-branch attraction" zone where
the tree with split $14|23$ tends to be inferred, and
\ifthenelse{\boolean{colour}}{red}{white} where $a\ge b$ are of small
to moderate size.

In our simulation studies parameters other than branch lengths 
were also varied.  Several of these govern
the details of the model of sequence evolution:
\begin{enumerate}
\item  Sequence length
\item The rates of insertions and deletions 
\item Parameters for the distribution of the size of indels
\end{enumerate}
Other parameters control the specifics of implementing our $k$-mer method:
\begin{enumerate}
\item[(4)] $k$, the $k$-mer length
\item[(5)]  $B$, the number of blocks
\end{enumerate}

For simulations that combine a site substitution process with an indel
process, one must specify the location of a root in the tree, since
indels change the sequence length; we chose the midpoint of the
interior branch to root the tree.  For initial sequence length at this
root, we chose the lengths $L = 1000, 10000$.  INDELible requires
users to choose a rate of insertion events and a rate of deletion
events, specified relative to the substitution rate; we set these
equal and denote the common value $\mu$.  In assuming that insertions
and deletions are rare relative to base substitutions, we varied this
parameter over the values $\mu = .01, .05, .1$.  We used the Lavalette
distribution as implemented in INDELible for determining the lengths
of inserted and deleted segments: For parameters $(a,M)$, this is the
distribution on $S=\{1,2, \ldots, M\}$ such that for $G \in S$, $Pr(G)
\propto \big( \frac{GM}{M - G +1} \big)^{-a}$.  Large $M$ and small
$a$ tend to produce longer insertion and deletion events.
\citet{Fletcher2009} suggest that values of $a \in[1.5,2]$ with a
large $M$ give a reasonable match with data.  We tried values $a =
1.1$ (as used in \citep{Chan2014}), $1.5$, $1.8$, and $M= 100$.

For testing our $k$-mer methods on simulated data, we varied $k = 1,3,5,7$ 
and the number of blocks $B$ ranged over $1,5, 25, 100, 250, 500$, provided this 
allowed a block size at least $k$.

\subsection*{Performance on simulated sequence data}

As presentation of all simulation results would require considerable
space, here we present only representative examples to illustrate key
points. The supplementary materials \citep{SuppMat} contain results of
other simulations\ifthenelse{\boolean{colour}}{}{, as well as color
  versions of the figures given here}.

\subsubsection*{Simulations with no  indel process.}

We begin by discussing simulations in which no indel process
occurs. This is the situation in which our theoretical results were
derived, and these runs investigate solely the effect of having
simulated sequence data of finite length.  These trials are, of
course, somewhat artificial in that in the absence of an indel process
we have exact alignments of sequences, and there is no reason to use
an alignment-free phylogenetic method.  Nonetheless, they represent a
measuring rod for evaluating the performance of the new methods
presented here.

We set the sequence length to 1000, and for comparison to traditional
approaches, produce Hulsenbeck diagrams in Figure
\ref{figure:nj_nogaps} using (i) the standard JC pairwise distance
formula for the sequences with the true alignment as produced by
INDELible together with Neighbor Joining (NJ), and (ii) the standard
JC distance formula after a pairwise alignment, followed by NJ.
Alignment in (ii) was performed by the Needleman-Wunsch algorithm
implemented in MATLAB's Bioinformatics Toolbox, but with scoring
parameters set to NCBI defaults: match$=2$, mismatch$=-3$, gap
existence$=-5$, gap extension$=-2$. Simulation (i) represents a
standard that would be desirable, but probably impossible, to match,
as $k$-mer methods make no use of the alignment itself and a true
alignment is never known in practice.  Simulation (ii) offers a more
realistic setting with results we might hope to match or beat, in
which large amounts of substitution results in quite dissimilar
sequences, and the introduction of gaps in the alignment process.  The
distance estimates computed with these `gappy' alignments can be quite
far from the true pairwise distances underlying the simulated data.

In Figure \ref{figure:nj_nogaps} (ii), for the simulation in which
sequences were aligned before distances were computed, there is a
rather pronounced region of parameter space to the upper left
displaying the phenomenon of long branch attraction.  In addition,
surrounding the \ifthenelse{\boolean{colour}}{red}{white} area where
the true tree is reliably constructed, we see a halo of
\ifthenelse{\boolean{colour}}{darkish blue}{dark grey}, illustrating
another region of parameter space with a weaker bias against the
correct tree.  Comparing (i) and (ii), it is clear that the alignment
process markedly degrades performance of the inference procedure.

\begin{figure}[h] 
\begin{center}
\ifthenelse{\boolean{colour}}{
 (i)
\resizebox{!}{3cm}{\includegraphics[height=3cm]{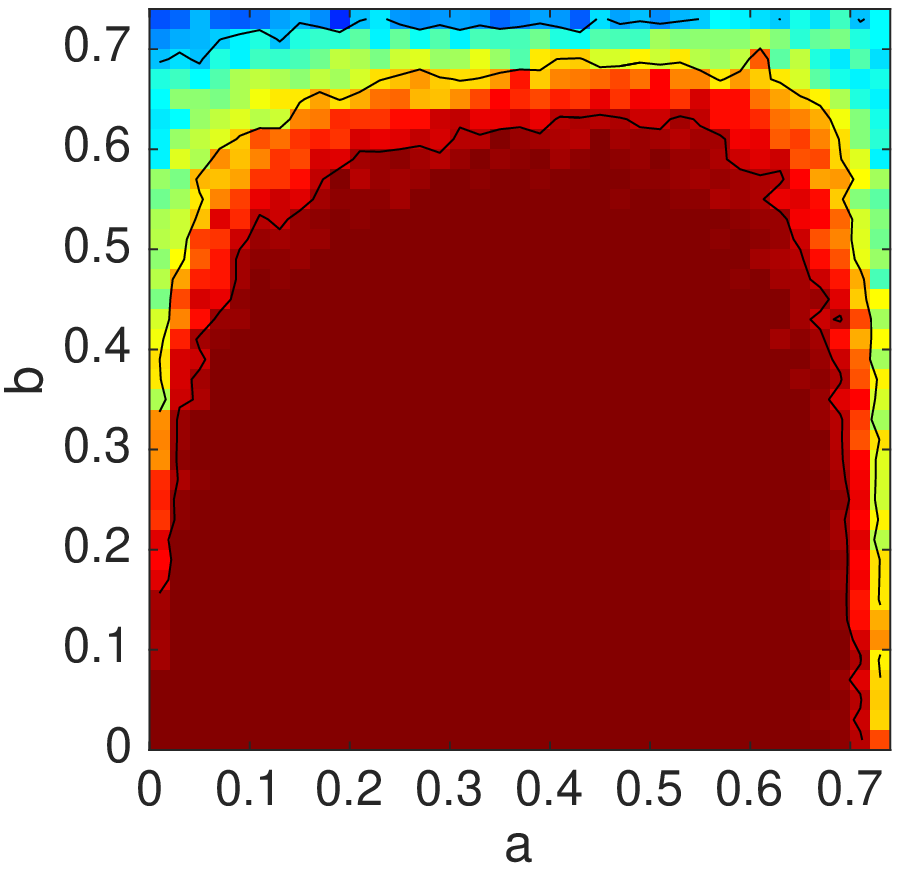}}
\quad (ii)
\resizebox{!}{3cm}{\includegraphics[height=3cm]{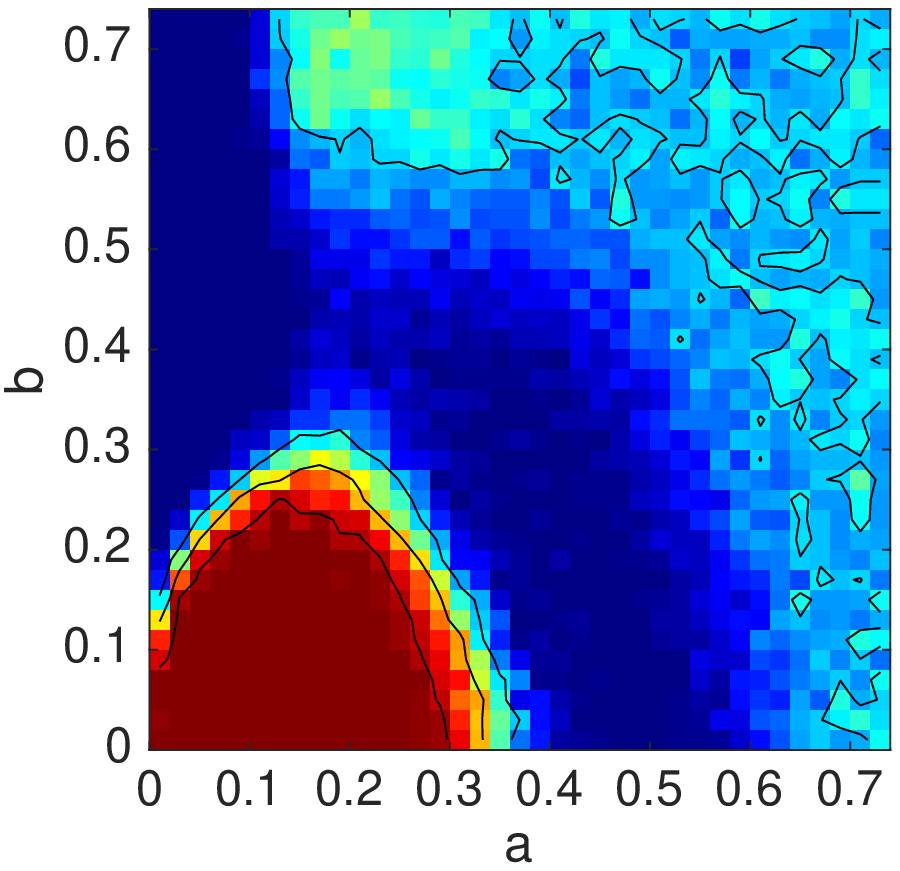}}
\, 
\resizebox{!}{!}{\includegraphics[height=3cm]{color_bar_C.eps}}
}{
 (i)
\resizebox{!}{3cm}{\includegraphics[height=3cm]{figure_6_2_true_align_nj_nogaps_1000_BW.eps}}
\quad (ii)
\resizebox{!}{3cm}{\includegraphics[height=3cm]{figure_6_2_align_nj_nogaps_1000_BW.eps}}
\, 
\resizebox{!}{!}{\includegraphics[height=3cm]{color_bar_BW.eps}}
}
\end{center}
\caption{\label{figure:nj_nogaps} Figures illustrating the accuracy of inference of tree 
topology on simulated data with no gaps, using the Jukes-Cantor distance and  
Neighbor Joining. Simulated sequences have length 1000 bp  with no indel process. 
In (i) the correct alignment is used, and in (ii) pairwise alignments are found before 
the JC distance is computed.}
\end{figure}

In Figure \ref{figure:kmer_nogaps_k5}, we present results using the
same simulated sequences (JC and no indels) as in the previous figure,
but use the distance $d_{JC}^{5,B}$ with NJ.  With $k=5$ held fixed,
we vary $B = 1, 5, 25, 100$.  This sequence of diagrams, in which the
\ifthenelse{\boolean{colour}}{red}{white} area increases with $B$,
illustrates that in the absence of indels and with $k$ held constant,
increasing the number of blocks is advantageous, as was anticipated in
Section \ref{sec:practical}.

Comparing Figure \ref{figure:kmer_nogaps_k5} with Figure
\ref{figure:nj_nogaps} (ii) suggests that when data sequences are
quite dissimilar, and a researcher might be inclined to align
sequences before a phylogenetic analysis, that our $k$-mer method can
outperform the traditional approach (alignment + $d_{JC}$ + NJ).  In
particular, using $d_{JC}^{k,B}$ the
\ifthenelse{\boolean{colour}}{red}{white} region of good performance
is enlarged, and the phenomenon of long branch attraction is
significantly lessened.  (Further simulations below will return to
this issue when there is a mild indel process.)

\begin{figure}[h] 
\begin{center} 
\ifthenelse{\boolean{colour}}
{
\resizebox{!}{3cm}{\includegraphics{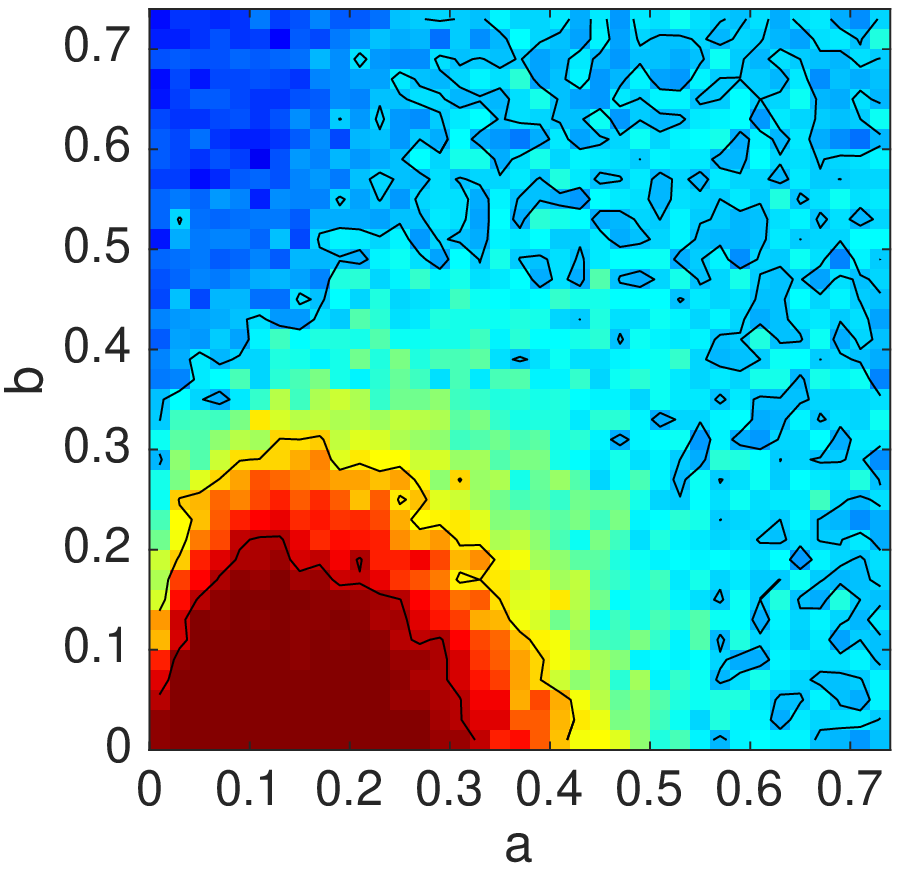}}
\resizebox{!}{3cm}{\includegraphics{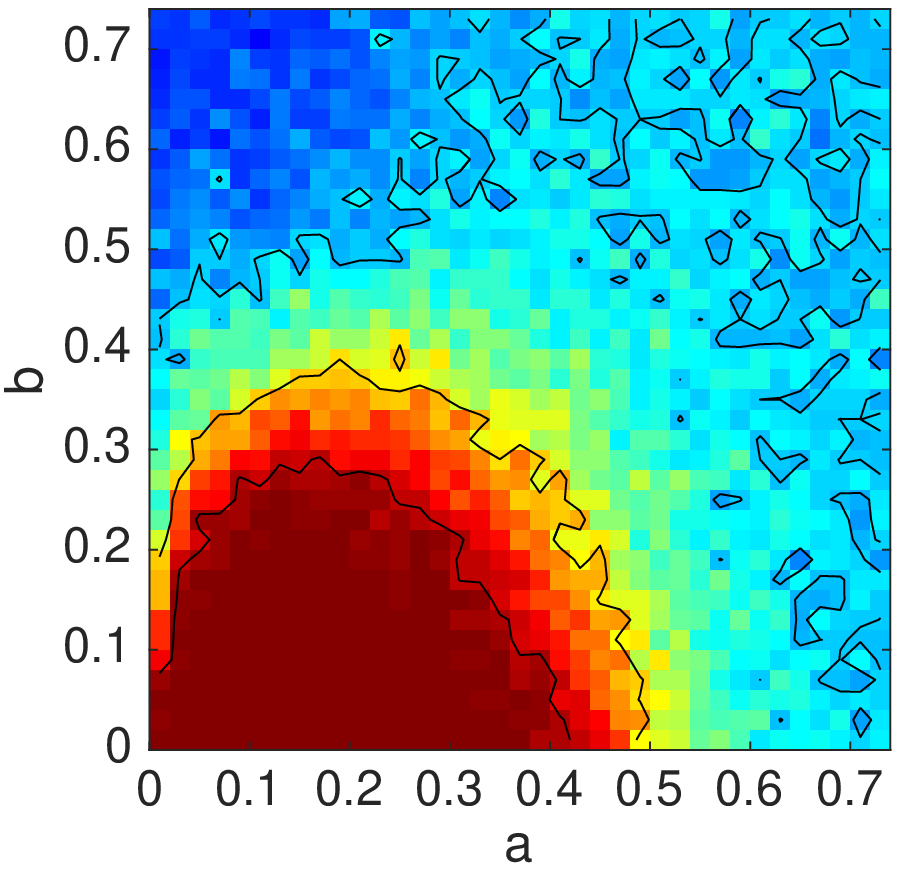}}
\resizebox{!}{3cm}{\includegraphics{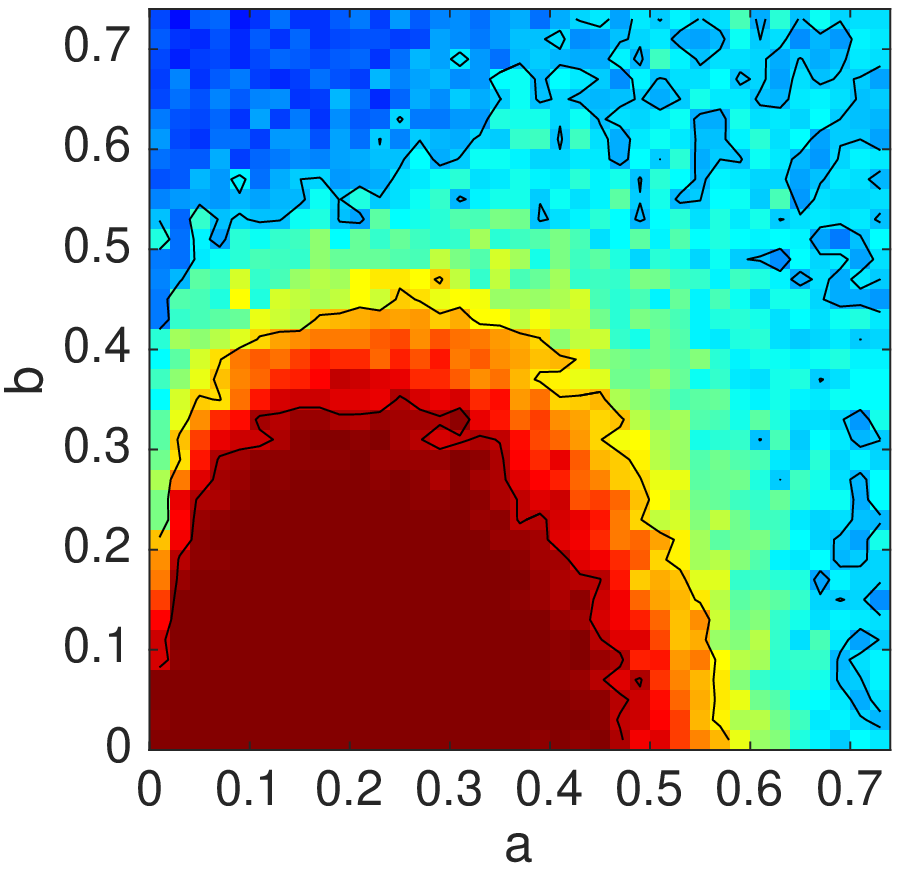}}
\resizebox{!}{3cm}{\includegraphics{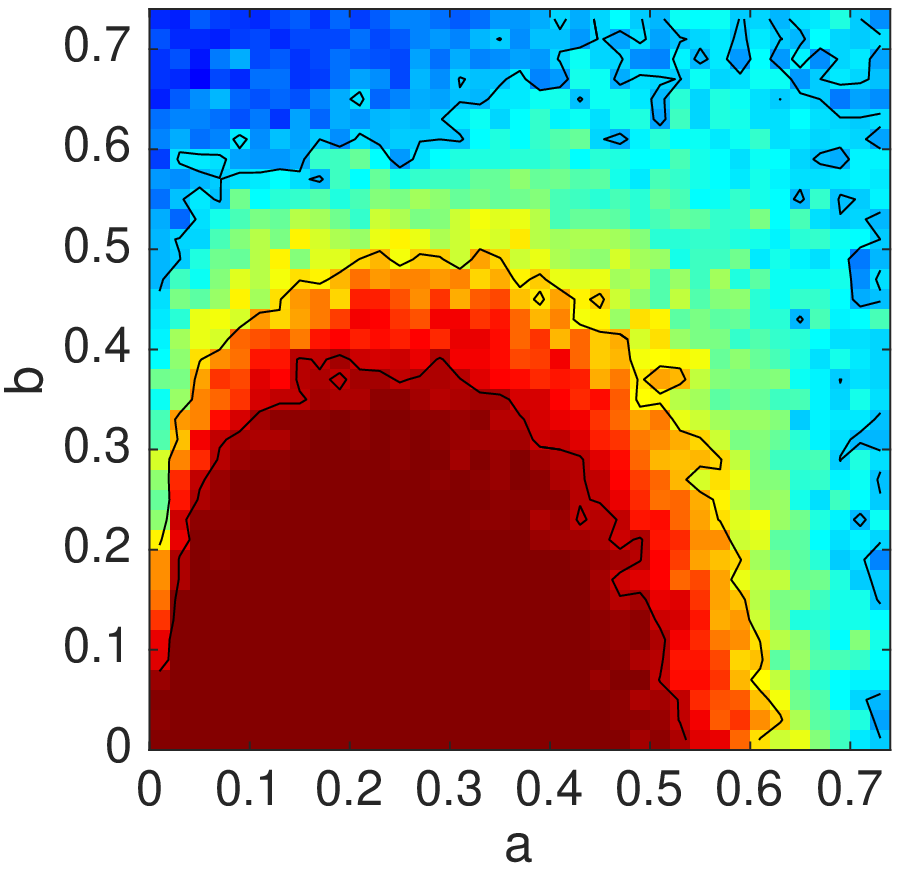}}
\resizebox{!}{!}{\includegraphics[height=3cm]{color_bar_C.eps}}
}{ 
\resizebox{!}{3cm}{\includegraphics{figure_6_3_k5_B1_L1000_nogaps_BW.eps}}
\resizebox{!}{3cm}{\includegraphics{figure_6_3_k5_B5_L1000_nogaps_BW.eps}}
\resizebox{!}{3cm}{\includegraphics{figure_6_3_k5_B25_L1000_nogaps_BW.eps}}
\resizebox{!}{3cm}{\includegraphics{figure_6_3_k5_B100_L1000_nogaps_BW.eps}}
\resizebox{!}{!}{\includegraphics[height=3cm]{color_bar_BW.eps}}
}
\end{center}
\caption{\label{figure:kmer_nogaps_k5}
Figures illustrating the accuracy of inference of 
tree topology on simulated data with no indels, using a 5-mer distance 
$d_{JC}^{5,B}$ and Neighbor Joining. Simulated sequences 
have length 1000 bp  with no indel process. From left to right, $B=1$, $5$, $25$, $100$.}
\end{figure}

Now fixing the number of blocks $B=25$, but varying $k$ in $d_{JC}^{k,25}$, with 
NJ we produce Figure \ref{figure:fixB_varyk_nogaps}. 
Notice here that with a fixed number of blocks, 
both too small and too large a value of $k$ reduces performance.

\begin{figure}[h]
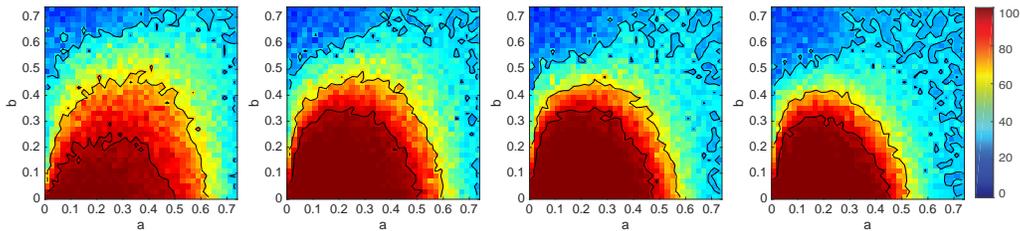
 
\begin{center} 
\ifthenelse{\boolean{colour}}{
\resizebox{!}{3cm}{\includegraphics{figure_6_4_k1_B25_L1000_nogaps.eps}}
\resizebox{!}{3cm}{\includegraphics{figure_6_4_k3_B25_L1000_nogaps.eps}}
\resizebox{!}{3cm}{\includegraphics{figure_6_4_k5_B25_L1000_nogaps.eps}}
\resizebox{!}{3cm}{\includegraphics{figure_6_4_k7_B25_L1000_nogaps.eps}}
\resizebox{!}{!}{\includegraphics[height=3cm]{color_bar_C.eps}}
}{
\resizebox{!}{3cm}{\includegraphics{figure_6_4_k1_B25_L1000_nogaps_BW.eps}}
\resizebox{!}{3cm}{\includegraphics{figure_6_4_k3_B25_L1000_nogaps_BW.eps}}
\resizebox{!}{3cm}{\includegraphics{figure_6_4_k5_B25_L1000_nogaps_BW.eps}}
\resizebox{!}{3cm}{\includegraphics{figure_6_4_k7_B25_L1000_nogaps_BW.eps}}
\resizebox{!}{!}{\includegraphics[height=3cm]{color_bar_BW.eps}}
}
\end{center}
\caption{\label{figure:fixB_varyk_nogaps} Figures illustrating the accuracy of 
inference of tree topology 
on data with no indels, using a $k$-mer distance $d_{JC}^{k,25}$
and Neighbor Joining.  Simulated sequences have length 1000 bp  with no indel process. 
From left to right, $k=1$, $3$, $5$, $7$.}
\end{figure}

In summary, while no performance of our $k$-mer distance comes close to the ideal  
of Figure \ref{figure:nj_nogaps} (i) (true alignment+$d_{JC}$+NJ), 
the $k$-mer methods often perform better than (alignment+$d_{JC}$+NJ) as shown in
Figure \ref{figure:nj_nogaps} (ii). 
Computing erroneous pairwise alignments results in a large region of parameter space in
which long branch attraction is pronounced, but
such biased inference is almost absent when $d_{JC}^{k,B}$ is used.
When the sequence length and number of blocks $B$ are fixed, the choice of $k$ can affect  performance,
with either too large or to small a $k$ causing degradation. 
It is unclear how to determine  an ``optimal'' choice of $k$ except through simulation.

\subsubsection*{Simulations with an indel process.}

With a length of 1000 bp for the sequence at the root of the tree, we now introduce an indel 
process with rate $\mu=.05$ and Lavalette parameters $a=1.8$, $M=100$. 
This means on average one insertion event and one deletion event occurs for every 
20 base substitutions.  
Repeating reconstruction methods (i) and (ii) of Figure \ref{figure:nj_nogaps}
on these datasets with indels, we obtain Figure \ref{figure:nj_gaps}.

\begin{figure}[h] 
\begin{center} 
\ifthenelse{\boolean{colour}}{ (i)
\resizebox{!}{3cm}{\includegraphics{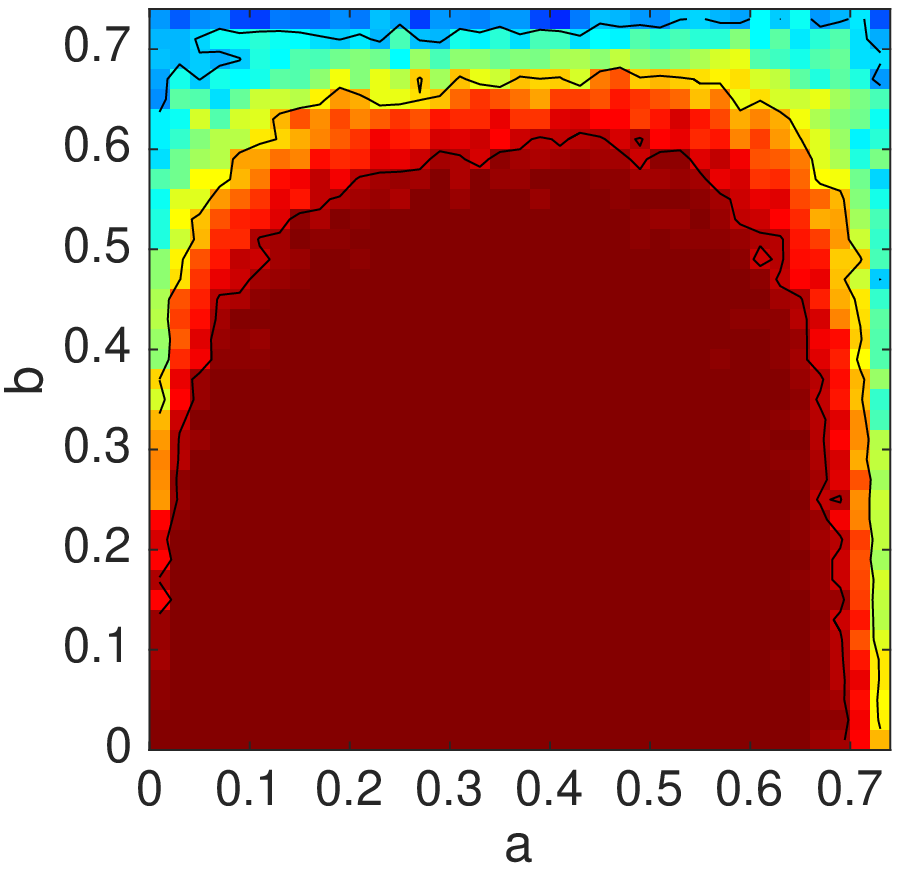}} 
\quad (ii)
\resizebox{!}{3cm}{\includegraphics{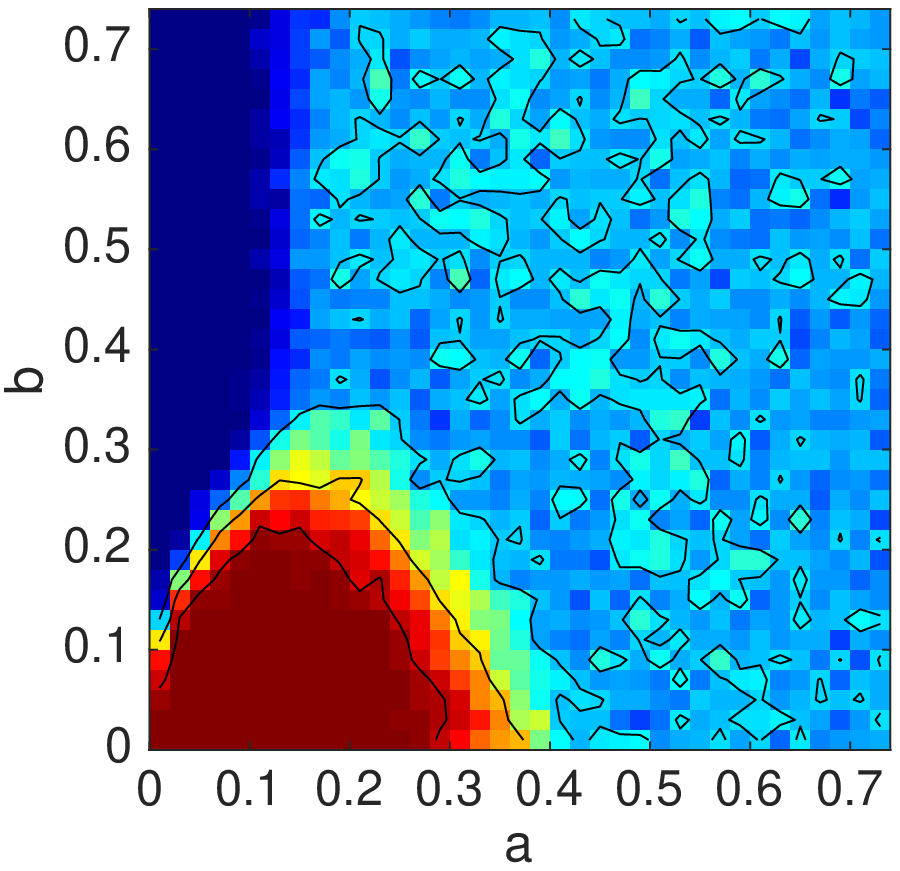}} 
\quad
\resizebox{!}{!}{\includegraphics[height=3cm]{color_bar_C.eps}}
}
{ (i)
\resizebox{!}{3cm}{\includegraphics{figure_6_5_tr_a1_8_idr0_05_gaps_BW.eps}} 
\quad (ii)
\resizebox{!}{3cm}{\includegraphics{figure_6_5_nj_a1_8_idr0_05_gaps_BW.eps}} 
\quad
\resizebox{!}{!}{\includegraphics[height=3cm]{color_bar_BW.eps}}
}
\end{center}
\caption{\label{figure:nj_gaps} Figures illustrating the accuracy of inference of tree topology 
on simulated data with indels using the Jukes-Cantor distance
and Neighbor Joining.  The root sequence is 1000 bp.  The indel process 
is determined by $\mu=.05$ and Lavalette parameters $a=1.8$, $M=100$. 
In (i) the true alignment is used, and in (ii) pairwise alignments are 
found before the JC distance is computed.}
\end{figure}

While Figure \ref{figure:nj_gaps} (i) shows excellent performance, it
assumes the correct alignment (including gaps) is known, which is
unrealistic in any empirical study.  Analysis (ii) is one that could
be performed on real data, and should be compared to Figure
\ref{figure:nj_nogaps} (ii) above.  For sequence data with indels the
region of good performance is similarly shaped, but smaller, than that
for data without indels.  This is to be expected, since even when few
substitutions occur, indels could lead to erroneous alignment.  In
both Figure \ref{figure:nj_nogaps} (ii) and Figure
\ref{figure:nj_gaps} (ii), long branch attraction is present in the
upper left corner of parameter space.  In contrast, however, in Figure
\ref{figure:nj_gaps} (ii) the area to the upper right surrounding the
area of good reconstruction does not display a bias against correct
reconstruction, but rather a uniform randomness in selection of the
tree.

Setting $k=5$ and $B = 1, 5, 25, 100$, and using $d_{JC}^{5,B}$+NJ on
the sequence data with indels produces Figure
\ref{figure:kmer_gaps_k5}.  Note that increasing the number of blocks
first improves performance, but then degrades it. This is explained by
a large number of blocks producing a small block size, which increases
the chance that corresponding blocks in two sequences share few
homologous sites, as was discussed in Section \ref{sec:practical}.
This phenomenon is only seen on data simulated with an indel process.
 
\begin{figure}[h]
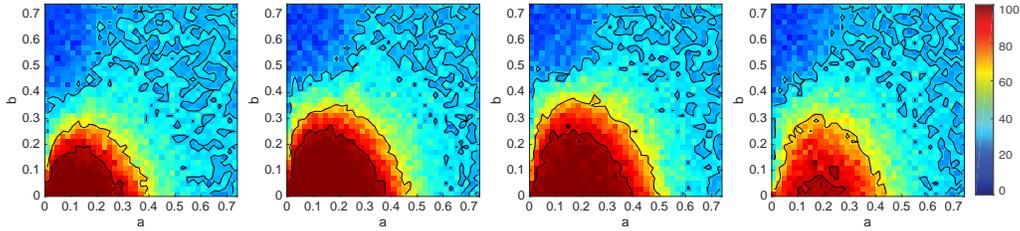
 
\begin{center} 
\ifthenelse{\boolean{colour}}
{
\resizebox{!}{3cm}{\includegraphics{figure_6_6_k5_B1_L1000_a18_idr05.eps}}
\resizebox{!}{3cm}{\includegraphics{figure_6_6_k5_B5_L1000_a18_idr05.eps}}
\resizebox{!}{3cm}{\includegraphics{figure_6_6_k5_B25_L1000_a18_idr05.eps}}
\resizebox{!}{3cm}{\includegraphics{figure_6_6_k5_B100_L1000_a18_idr05.eps}}
\resizebox{!}{!}{\includegraphics[height=3cm]{color_bar_C.eps}}
}{
\resizebox{!}{3cm}{\includegraphics{figure_6_6_k5_B1_L1000_a18_idr05_BW.eps}}
\resizebox{!}{3cm}{\includegraphics{figure_6_6_k5_B5_L1000_a18_idr05_BW.eps}}
\resizebox{!}{3cm}{\includegraphics{figure_6_6_k5_B25_L1000_a18_idr05_BW.eps}}
\resizebox{!}{3cm}{\includegraphics{figure_6_6_k5_B100_L1000_a18_idr05_BW.eps}}
\resizebox{!}{!}{\includegraphics[height=3cm]{color_bar_BW.eps}}
}
\end{center}
\caption{\label{figure:kmer_gaps_k5} Figures illustrating the accuracy of 
inference of tree topology on simulated data with 
indels, using a 5-mer distance $d_{JC}^{5,B}$ and Neighbor Joining. 
The root sequence is 1000 bp.  The indel process 
is determined by $\mu=.05$ and Lavalette parameters $a=1.8$, $M=100$. 
From left to right, $B=1$, $5$, $25$, $100$.}
\end{figure}

With the number of blocks set at 25, but varying $k$ in $d_{JC}^{k,25}$, we obtain 
Figure \ref{figure:fixB_varyk_gaps}. 
Again we note that for a fixed number of blocks, too small or large a value of
$k$ degrades performance.

\begin{figure}[h]
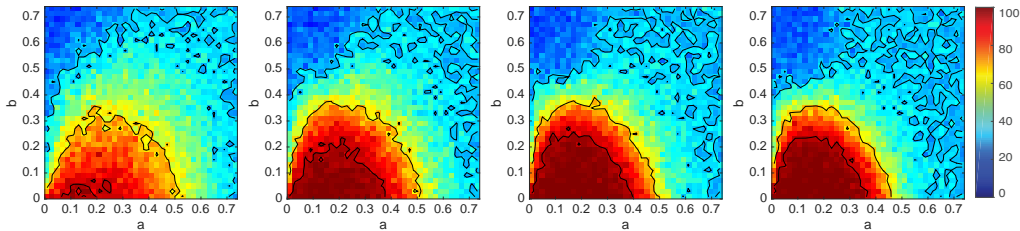
 
\begin{center} 
\ifthenelse{\boolean{colour}}
{
\resizebox{!}{3cm}{\includegraphics{figure_6_7_k1_B25_L1000_a18_idr05_gaps.eps}}
\resizebox{!}{3cm}{\includegraphics{figure_6_7_k3_B25_L1000_a18_idr05_gaps.eps}}
\resizebox{!}{3cm}{\includegraphics{figure_6_7_k5_B25_L1000_a18_idr05_gaps.eps}}
\resizebox{!}{3cm}{\includegraphics{figure_6_7_k7_B25_L1000_a18_idr05_gaps.eps}}
\resizebox{!}{!}{\includegraphics[height=3cm]{color_bar_C.eps}}
}{
\resizebox{!}{3cm}{\includegraphics{figure_6_7_k1_B25_L1000_a18_idr05_gaps_BW.eps}}
\resizebox{!}{3cm}{\includegraphics{figure_6_7_k3_B25_L1000_a18_idr05_gaps_BW.eps}}
\resizebox{!}{3cm}{\includegraphics{figure_6_7_k5_B25_L1000_a18_idr05_gaps_BW.eps}}
\resizebox{!}{3cm}{\includegraphics{figure_6_7_k7_B25_L1000_a18_idr05_gaps_BW.eps}}
\resizebox{!}{!}{\includegraphics[height=3cm]{color_bar_BW.eps}}
}
\end{center}
\caption{\label{figure:fixB_varyk_gaps} Figures illustrating the accuracy of 
inference of tree topology on simulated data with indels
using a k-mer distance $d_{JC}^{k,25}$ and Neighbor Joining. 
The root sequence is 1000 bp.  The indel process 
is determined by $\mu=.05$ and Lavalette parameters $a=1.8$, $M=100$. 
From left to right, $k =1, 3, 5, 7$.}
\end{figure}

These figures illustrate that even in the presence of a mild indel
process, the $k$-mer method described here can perform as well as
pairwise alignment with traditional distance methods in the regions of
parameter space where those work well, yet greatly reduce the
pronounced long-branch attraction problems that incorrect alignements
introduce in other regions of parameter space.  Although the $k$-mer
distance $d_{JC}^{k,B}$ was derived using a model with no indels,
these simulations demonstrate its performance is somewhat robust to
violation of that assumption.

\subsubsection*{Other $k$-mer methods}

To conclude, in Figure \ref{figure:other_kmer} we display some
diagrams that illustrate the performance of other $k$-mer distance
methods \citep{Vinga2003, Chan2014, Reinert2009, Wan2010, Edgar2004b}
on simulated data with indels.  The datasets were the same ones used
in producing Figures \ref{figure:nj_gaps}, \ref{figure:kmer_gaps_k5},
\ref{figure:fixB_varyk_gaps}.

In the figure below, we use $k$-mer distances previously proposed:
With $x_1$ and $x_2$ the observed $k$-mer count vectors, two of these
distances are
\begin{align}
L_2^2&=\|x_1-x_2\|^2_2\label{eq:L22}\\
{\theta}&= \arccos(x_1\cdot x_2/ \|x_1\|_2\|x_2\|_2)\label{eq:angle}
\end{align}
These have long been studied for sequence comparison \citep{Vinga2003}, 
though primarily for non-phylogenetic applications.  \citet{Yang2008} used
a variation of the $L_2^2$ distance based on replacing $x_1$ and $x_2$ with
$x_1/(n_1 - k +1)$ and $x_2/(n_2 - k+1)$, respectively, where $n_i$ is the length of sequence $i$.

The next three have appeared in phylogenetic investigations of \citet{Chan2014}, 
but are based on sequence comparison methods developed for other purposes, as 
reviewed by \citet{Song2013}. 
With $\tilde x_i=x_i-\EE(x_i)$ the centralized count vector, let
\begin{align*}
D_2(x_1, x_2) &=x_1\cdot x_2,\\
D_2^S(x_1,x_2)&=\sum_W \frac {\tilde x_1^W \tilde x_2^W} {\sqrt{(\tilde x_1^W)^2 +(\tilde x_2^W)^2}},\\
D_2^*(x_i,x_2)&=\sum_W \frac {\tilde x_1^W \tilde x_2^W}{\sqrt{\EE(x_1^W)\EE(x_2^W)}},
\end{align*}
as did  \citet{Reinert2009} and \citet{Wan2010}. Then define the distances
\begin{align}
d_2&= \left | \ln \frac {D_2(x_1,x_2)} {\sqrt{D_2(x_1,x_1)D_2(x_2,x_2)} } \right | ,\label{eq:d2}\\
d_2^S&= \left | \ln \frac {D_2^S(x_1,x_2)} {\sqrt{D_2^S(x_1,x_1)D_2^S(x_2,x_2)} } \right | ,\label{eq:d2S}\\
d_2^*&= \left | \ln \frac {D_2^*(x_1,x_2)} {\sqrt{D_2^*(x_1,x_1)D_2^*(x_2,x_2)} } \right |,\label{eq:d2*}
\end{align}
with the convention that the logarithm of a negative number is set to
$\infty$.  As the distances $d_2$ and $\theta$ differ from each other
by the application of a monotone function, for $4$-leaf trees they
perform identically using UPGMA, and quite similarly with NJ.  Thus,
in Figure \ref{figure:other_kmer} the plot for the $\theta$ distance
is not shown.

Finally, for comparison purposes, we include the distance used in the
initial step of the MUSCLE alignment algorithm \citep{Edgar2004b},
\begin{equation}
 m = 1 - \sum_W \frac{\min\{x_1^W, x_2^W\}}{(n-k+1)},\label{eq:MUSCLE}
 \end{equation}
 where $n=\min(n_1,n_2)$ is the length of the shorter of the two
 sequences. Since MUSCLE uses UPGMA as its default for tree building,
 we performed both NJ and UPGMA for all of these distances.

As is apparent in in Figure \ref{figure:other_kmer}, with $k=5$ most
of the distances in (\ref{eq:L22}-\ref{eq:MUSCLE}) exhibit long branch
attraction bias which is generally quite pronounced, and fail to match
the performance of the $5$-mer distance derived here.

\begin{figure}[h]
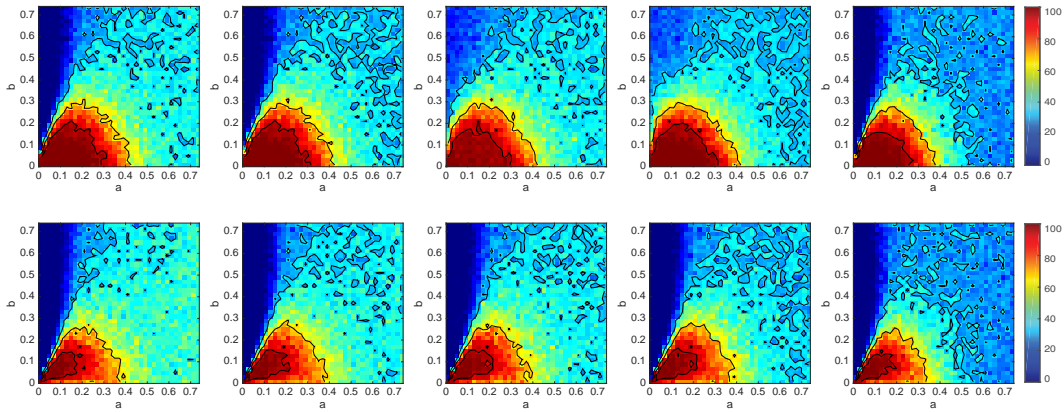
 
\begin{center} 
\ifthenelse{\boolean{colour}}
{
\resizebox{!}{2.5cm}{\includegraphics{figure_6_8_k5_L1000_a18_idr05_gaps_ltwo_nj.eps}} 
\resizebox{!}{2.5cm}{\includegraphics{figure_6_8_k5_L1000_a18_idr05_gaps_d2_nj.eps}} 
\resizebox{!}{2.5cm}{\includegraphics{figure_6_8_k5_L1000_a18_idr05_gaps_d2s_nj.eps}} 
\resizebox{!}{2.5cm}{\includegraphics{figure_6_8_k5_L1000_a18_idr05_gaps_d2star_nj.eps}} 
\resizebox{!}{2.5cm}{\includegraphics{figure_6_8_k5_L1000_a18_idr05_gaps_mus_nj.eps}} 
\resizebox{!}{!}{\includegraphics[height=2.5cm]{color_bar_C.eps}}

\bigskip

\resizebox{!}{2.5cm}{\includegraphics{figure_6_8_k5_L1000_a18_idr05_gaps_ltwo_upgma.eps}} 
\resizebox{!}{2.5cm}{\includegraphics{figure_6_8_k5_L1000_a18_idr05_gaps_d2_upgma.eps}} 
\resizebox{!}{2.5cm}{\includegraphics{figure_6_8_k5_L1000_a18_idr05_gaps_d2s_upgma.eps}} 
\resizebox{!}{2.5cm}{\includegraphics{figure_6_8_k5_L1000_a18_idr05_gaps_d2star_upgma.eps}}
\resizebox{!}{2.5cm}{\includegraphics{figure_6_8_k5_L1000_a18_idr05_gaps_mus_upgma.eps}}  
\resizebox{!}{!}{\includegraphics[height=2.5cm]{color_bar_C.eps}}
}{
\resizebox{!}{2.5cm}{\includegraphics{figure_6_8_k5_L1000_a18_idr05_gaps_ltwo_nj_BW.eps}} 
\resizebox{!}{2.5cm}{\includegraphics{figure_6_8_k5_L1000_a18_idr05_gaps_d2_nj_BW.eps}} 
\resizebox{!}{2.5cm}{\includegraphics{figure_6_8_k5_L1000_a18_idr05_gaps_d2s_nj_BW.eps}} 
\resizebox{!}{2.5cm}{\includegraphics{figure_6_8_k5_L1000_a18_idr05_gaps_d2star_nj_BW.eps}} 
\resizebox{!}{!}{\includegraphics[height=2.5cm]{color_bar_BW.eps}}
\resizebox{!}{2.5cm}{\includegraphics{figure_6_8_k5_L1000_a18_idr05_gaps_mus_nj_BW.eps}} 

\bigskip

\resizebox{!}{2.5cm}{\includegraphics{figure_6_8_k5_L1000_a18_idr05_gaps_ltwo_upgma_BW.eps}} 
\resizebox{!}{2.5cm}{\includegraphics{figure_6_8_k5_L1000_a18_idr05_gaps_d2_upgma_BW.eps}} 
\resizebox{!}{2.5cm}{\includegraphics{figure_6_8_k5_L1000_a18_idr05_gaps_d2s_upgma_BW.eps}} 
\resizebox{!}{2.5cm}{\includegraphics{figure_6_8_k5_L1000_a18_idr05_gaps_d2star_upgma_BW.eps}} 
\resizebox{!}{2.5cm}{\includegraphics{figure_6_8_k5_L1000_a18_idr05_gaps_mus_upgma_BW.eps}} 
\resizebox{!}{!}{\includegraphics[height=2.5cm]{color_bar_BW.eps}}
}
\end{center}
\caption{\label{figure:other_kmer} Figures illustrating the accuracy of 
inference of tree topology on simulated data with 
indels, using a variety of distances and Neighbor Joining and UPGMA.
The root sequence is 1000 bp.  The indel process 
is determined by $\mu=.05$ and Lavalette parameters $a=1.8$, $M=100$. 
Columns in the figure are, from left to right, obtained using the distances given in 
equations \eqref{eq:L22},  \eqref{eq:d2}, \eqref{eq:d2S}, \eqref{eq:d2*}, and  \eqref{eq:MUSCLE}, all with $k=5$.
The top row of figures uses Neighbor Joining, and the bottom UPGMA.}
\end{figure}

%%%%%%%%%%%%%%%%%%%%%%%%%%%%%%%%%%
%%%%%%%%%%%%%%%%%%%%%%%%%%%%%%%%%%
%%%%%%%%%%%%%%%%%%%%%%%%%%%%%%%%%%
%%%%%%%%%%%%%%%%%%%%%%%%%%%%%%%%%%

\section{Conclusions}

We have derived model-based distance corrections for the
squared-Euclidean distance between $k$-mer count vectors of sequences.
Our results show that the uncorrected use of the squared-Euclidean
distance leads to statistically inconsistent estimation of the tree
topology, with inherent long-branch attraction problems.  This
statistical inconsistency occurs even at short branch lengths, and is
strongly manifested in simulations.  Simulations show that our
corrected distance outperforms previously proposed $k$-mer methods,
and suggest that many of those are statistically inconsistent with
long-branch attraction biases.

All our results have been derived under the assumption that there are
no insertions or deletions in the evolution of sequences.  Our
simulations indicate that even if a mild indel process occured, a
simple extension of the corrected method still performs well.  It
remains to develop $k$-mer methods assuming an indel process, using
the indel model structure to develop a more precise correction on the
distance.

\citet{Daskalakis2013} developed an alignment-free phylogenetic tree
inference method for a model with a simple indel process.  Their
method can be seen as a $1$-mer method.  While we have not compared
their method directly to any of ours, our simulations suggest that
$1$-mer methods perform poorly compared to $k$-mer methods with larger
$k$.  This suggests that a natural line for future research would be
to combine the approach of Daskalakis and Roch with ours to develop
consistent $k$-mer methods that take into account the structure of an
underlying indel model.

\section*{Acknowledgement}
\label{sec:acknowledgement}
Seth Sullivant was partially supported by the David and Lucille
Packard Foundation and the US National Science Foundation (DMS
0954865).  All the authors thank UAF's Arctic Regional Supercomputing
Center and their staff for use of their cluster and help with parallel
implementation of the simulations.

\bibliography{kmer}
\bibliographystyle{plainnat}

%%%%%%%%%%%%%%%%%%%%%%%%%%%%%%%%%%
%%%%%%%%%%%%%%%%%%%%%%%%%%%%%%%%%%
%%%%%%%%%%%%%%%%%%%%%%%%%%%%%%%%%%
%%%%%%%%%%%%%%%%%%%%%%%%%%%%%%%%%%

\section*{Appendix A:  Proofs for Section \ref{sec:formulas}} \label{app:A}

Here we establish Theorem \ref{thm:gentrace}.
\smallskip

Recalling notation from Section  \ref{sec:formulas}, $\pi=(\pi^w)_{w\in L}$ is the stationary distribution for $M$, an  $L\times L$ Markov matrix describing the single-site state change process from sequence $S_1$ to sequence $S_2$. The probability of a $k$-mer $W=w_1w_2\ldots w_k\in [L]^k$ in any $k$ consecutive sites of either single sequence is $\pi^{W} =\prod_{j=1}^k \pi^{w_j}$.

Then $P=\diag(\pi)M$ is the joint distribution of states in aligned sites of the two sequences. We can alternately  view the state changes  from $S_2$ to $S_1$ as described by the Markov matrix  $N=\diag(\pi)^{-1}P^{\operatorname{T}}$, where $\operatorname{T}$ denotes transpose. For future use, note that $\tr M=\tr N$, where $\tr$ denotes the trace.

Let $X_{\ell i}^W$ be an indicator variable for the occurrence of a $k$-mer $W$ in sequence $\ell=1,2$ starting at position $i$.
Then $X_\ell^W=\sum_{i=1}^{n-k+1} X_{\ell i}^W$ is the count of occurrences of $k$-mer $W$ in sequence $\ell$, and
$X_{\ell}=(X_\ell^ W)_{W\in [L]^k}$ is the random vector of $k$-mer counts in the sequence.

Let $Z_i^W=X_{1i}^W-X_{2i}^W$ and $ Z^W=\sum_{i=1}^{n-k+1} Z_i^W=X_1^W-X_2^W$. These random variables have mean 0.

\begin{prop} \label{prop:cov} 
For  $i\ne j$,
$$\sum_{W\in[L]^k} \frac 1{\pi^W} \Cov[Z_i^W,Z_j^W]=0.$$
\end{prop}
\begin{proof} We may assume $i<j$.
If $j-i\ge k$, the variables $Z_i^W$, $Z_j^W$ depend on disjoint sets of sites, hence are independent.
This is all that is needed for $k=1$.

We now proceed by induction, assuming the result holds for $(k-1)$-mers, and considering only cases with $0<j-i<k$.
Writing a $k$-mer $W$ as a $(k-1)$-mer $W'$ followed by a 1-mer $w$, so that $\pi^W=\pi^{W'}\pi^{w}$, we have

\begin{align}
\sum_{W\in[L]^k} \frac 1{\pi^W} \Cov[Z_i^W,Z_j^W]&=\notag\\
\sum_{W'\in[L]^{k-1}}\frac 1{\pi^{W'}} & \sum_{w\in [L]} \frac1{\pi^{w}} \EE[Z_i^{W'w}Z_j^{W'w}]\notag\\
=\sum_{W'\in[L]^{k-1}}\frac 1{\pi^{W'}} &\sum_{w\in [L]} \frac1{\pi^{w}} \EE[X_{1i}^{W'w}X_{1j}^{W'w}-X_{1i}^{W'w}X_{2j}^{W'w} -X_{2i}^{W'w}X_{1j}^{W'w}+X_{2i}^{W'w}X_{2j}^{W'w}]\label{eq:cov}
\end{align}

Now since $j-i<k$,
$$X_{1i}^{W'w}X_{1j}^{W'w}=X_{1i}^{W'} X_{1j}^{W'}X_{1(i+k-1)}^w X_{1(j+k-1)}^w$$ 
and 
$$\EE[X_{1i}^{W'w}X_{1j}^{W'w} ]=\EE[X_{1i}^{W'} X_{1j}^{W'}] \delta(u,w)\pi^w,$$
where $u=w_{k-j+i}$ is the  $(k-j+i)$th letter in $W$, and $\delta(u,w)$ is the Kronecker delta.
Thus
$$ \sum_{w\in [L]} \frac1{\pi^{w}} \EE[X_{1i}^{W'w}X_{1j}^{W'w}]=\sum_{w\in [L]}  \EE[X_{1i}^{W'}X_{1j}^{W'}]\delta(u,w)=  \EE[X_{1i}^{W'}X_{1j}^{W'}].
$$
Likewise, $ \sum_{w\in [L]} \frac1{\pi^{w}} \EE[X_{2i}^{W'w}X_{2j}^{W'w}]=  \EE[X_{2i}^{W'}X_{2j}^{W'}]
$.

In a similar way we see
$$X_{2i}^{W'w}X_{1j}^{W'w}=X_{2i}^{W'} X_{1j}^{W'}X_{2(i+k-1)}^wX_{1(j+k-1)}^w$$ 
and
$$\EE[ X_{2i}^{W'w}X_{1j}^{W'w}]=\EE[X_{2i}^{W'} X_{1j}^{W'}]M(u,w) \pi^w ,$$
where $u=w_{k-j+i}$ is the  $(k-j+i)$th letter in $W$, and $M$ is the Markov matrix describing the substitution process from sequence 1 to sequence 2.
Thus
$$\sum_{w\in [L]} \frac1{\pi^{w}} \EE[X_{2i}^{W'w}X_{1j}^{W'w}]=\sum_{w\in [L]}  \EE[X_{2i}^{W'}X_{1j}^{W'}]M(u,w)=  \EE[X_{2i}^{W'}X_{1j}^{W'}]$$
and, similarly, $ \sum_{w\in [L]} \frac1{\pi^{w}}\EE[ X_{1i}^{W'w}X_{2j}^{W'w}]
=\EE[X_{1i}^{W'} X_{2j}^{W'}].$

Combining these expected values with equation \eqref{eq:cov} we have
\begin{align*}
\sum_{W\in[L]^k} \frac 1{\pi^W} \Cov[Z_i^W,Z_j^W]&=\sum_{W'\in[L]^{k-1}}\frac 1{\pi^{W'}}  \EE[X_{1i}^{W'}X_{1j}^{W'}-X_{1i}^{W'}X_{2j}^{W'} -X_{2i}^{W'}X_{1j}^{W'}+X_{2i}^{W'}X_{2j}^{W'}]\\
&=\sum_{W'\in[L]^{k-1}}\frac 1{\pi^{W'}}  \Cov[Z_i^{W'},Z_j^{W'}] =0
\end{align*}
by the inductive hypothesis.
\end{proof}

\begin{proof}[Proof of Theorem \ref{thm:gentrace}]
For $k=1$, using Proposition \ref{prop:joint} we have
\begin{align*} \mathbb{E}\left[  \sum_w  \frac1{\pi^ w} (X_1^ w - X_{2}^{ w}) ^{2} \right]&=
  \sum_w {\pi^ w} \EE\left [ \left (\frac {X_1^ w}{\pi^w} - \frac{X_{2}^{ w}}{\pi^w}\right) ^{2}\right ] \\
&=\sum_w \pi^w \frac n{(\pi^w)^2}2(\pi^w -P_{ww})\\ 
&=\sum_w 2n(1 -M_{ww})\\ 
&=2n(L-\tr M).
\end{align*}

Now inductively suppose the result holds for $(k-1)$-mers, and consider $k$-mers.
Then, since  $Z^W$ has mean zero, 
\begin{align*}\EE\left[  \sum_W  \frac1{\pi^W} (X_{1}^{W} - X_{2}^{W}) ^{2} \right]  &=
  \sum_W  \frac1{\pi^W} \EE \left [(Z^W) ^{2} \right]  \\
  &= \sum_W  \frac1{\pi^W} \Var\left [Z^{W}\right ]  \\
  &= \sum_W  \frac1{\pi^W}  \Var \left [\sum_i Z_i^{W} \right]  \\
  &=\sum_W  \frac1{\pi^W} \left (  \sum_i \Var\left[Z_i^W\right ] + \sum_{i\ne j}\Cov\left [Z_i^W,Z_j^W\right ]    \right )\\
  &=\sum_W  \frac1{\pi^W} \left (  \sum_i \Var\left [Z_i^W\right]    \right ).
\end{align*}
Here Proposition \ref{prop:cov} justifies the last equality.
Now since $(Z_i^W)^2$ is the indicator variable for when exactly one of $X_{1i}^W$, $X_{2i}^W$ is 1, 
$$ \Var[Z_i^W] = \EE[(Z_1^W)^2]= \pi^W\left(1-\prod_{j=1}^kM(w_j,w_j) \right ) +  \pi^W\left(1-\prod_{j=1}^kN(w_j,w_j) \right ).$$
Thus
\begin{align*}
\mathbb{E}\left[  \sum_W  \frac1{\pi^W} (X_{1}^{W} - X_{2}^{W}) ^{2} \right]  & = (n-k+1)\sum_W \left ( 2- \prod_{j=1}^kM(w_j,w_j) -\prod_{j=1}^kN(w_j,w_j) \right )\\
&=(n-k+1)( 2L^k -(\tr M)^k -(\tr N)^k)\\
&=2(n-k+1)(L^k-(\tr M)^k ).
\end{align*}
\end{proof}

%%%%%%%%%%%%%%%%%%%%%%%%%%%%%%%%%%%%%%%%%%%%%%%%
%%%%%%%%%%%%%%%%%%%%%%%%%%%%%%%%%%%%%%%%%%%%%%%%
%%%%%%%%%%%%%%%%%%%%%%%%%%%%%%%%%%%%%%%%%%%%%%%%

\section*{Appendix B: Proofs for Section \ref{sec:ident} }

We establish Proposition \ref{prop:joint}. Our proof is independent of earlier arguments, as the result is needed in Appendix A.

Sequences $S_1$ and $S_2$ each have i.i.d. sites with state probabilities given by $\pi_1$ and 
$\pi_2$, and site transition probabilities from $S_1$ to $S_2$ are given by the matrix $M$. 
Note that the $\pi_\ell$ need not be stationary vectors for $M$.

As in Appendix A, define random variables $X_{\ell k}^w$ for $w\in[L]$, $\ell\in\{1,2\}$, $j\in[n]$ to be 
indicators of state $w$ in sequence $\ell$ at site $j$.
The 1-mer distribution vector for the sequence $S_\ell$ is then
$X_\ell$ with entries $X_\ell^w=\sum_{j=1}^n X_{\ell j}^w.$ 

\begin{proof}[Proof of Proposition \ref{prop:joint}]
That $\pi_\ell =\frac 1n \EE(X_\ell )$ is clear.

Since $\EE\left [\frac {X_1^u} {\pi_1^u}-\frac {X_2^w  }{\pi_2^w}\right ]=0$,
\begin{align}
\EE\left [ \left (\frac {X_1^u} {\pi_1^u}-\frac {X_2^w  }{\pi_2^w}\right )^2\right ]&=\Var\left [ \frac {X_1^u} {\pi_1^u}-\frac {X_2^w  }{\pi_2^w} \right ]= \Var\left [ \sum_{j=1}^n \left (\frac {X_{1j}^u} {\pi_1^u}-\frac {X_{2j}^w  }{\pi_2^w} \right )\right ]\notag \\
&= \sum_{j=1}^n \Var \left [\frac {X_{1j}^u} {\pi_1^u}-\frac {X_{2j}^w  }{\pi_2^w} \right ] =n \cdot\Var \left [\frac {X_{11}^u} {\pi_1^u}-\frac {X_{21}^w  }{\pi_2^w} \right ]\label{eq:var1}
\end{align}
by the i.i.d.~assumption.
But 
\begin{align*}
\Var \left [\frac {X_{11}^u} {\pi_1^u}-\frac {X_{21}^w  }{\pi_2^w} \right ]&=\EE \left [\left (\frac {X_{11}^u} {\pi_1^u}-\frac {X_{21}^w  }{\pi_2^w} \right )^2\right ]=\EE \left [ \frac {(X_{11}^u)^2 }{(\pi_1^u)^2}+\frac {(X_{21}^w )^2 }{(\pi_2^w)^2} -2\frac {X_{11}^uX_{21}^w}{\pi_1^u\pi_2^w}\right ] \\
&=\EE \left [ \frac {X_{11}^u }{(\pi_1^u)^2}+\frac {X_{21}^w  }{(\pi_2^w)^2} -2\frac {X_{11}^uX_{21}^w}{\pi_1^u\pi_2^w}\right ].
\end{align*}
Since $\EE[X_{\ell 1}^u]=\pi_\ell ^u$ and $\EE [X_{11}^uX_{21}^w]=P_{uw}$, this shows
\begin{equation}
\Var \left [\frac {X_{11}^u} {\pi_1^u}-\frac {X_{21}^w  }{\pi_2^w} \right ]= \frac 1{\pi_1^u} +\frac 1{\pi_2^w}-2\frac {P_{uw}}{\pi_1^u\pi_2^w}.\label{eq:var2}\end{equation}
Substituting equation \eqref{eq:var2} into equation \eqref{eq:var1} and solving for $P_{uw}$ completes the proof.\end{proof}

\medskip
To establish Proposition \ref{prop:joint3}, recall that for $\ell=1,2,3$ we consider sequences $S_\ell$ with $1$-mer count vectors $X_\ell$ and base distribution vector $\pi_\ell=\EE(X_\ell)$. Let
$$Y_{uvw}=\alpha {X_1^u} +\beta X_2^v  +\gamma X_3^w,$$
where $\alpha, \beta,\gamma$ are constants chosen so
$$ \alpha{\pi_1^u}+ \beta{\pi_2^v}+  \gamma{\pi_3^w}=0.$$

\begin{proof}[Proof of Proposition \ref{prop:joint3}]
Note $\EE(Y_{uvw})=0$. Using the fact that the 3rd central moment is additive over independent variables, and that sites are identically distributed
\begin{align*}
\EE\left (Y_{uvw}^3 \right )&= \EE\left ( \left ( \sum_{i=1}^n \left (\alpha  {X_{1i}^u} +\beta  {X_{2i}^v  }+\gamma  {X_{3i}^w} \right )\right )^3 \right )\notag \\
&= \sum_{i=1}^n \EE\left ( \left ( \alpha  {X_{1i}^u} +\beta  {X_{2i}^v  }+\gamma  {X_{3i}^w}\right )^3 \right )=n \cdot\EE\left ( \left ( \alpha  {X_{11}^u} +\beta  {X_{21}^v  }+\gamma  {X_{31}^w} \right )^3 \right ),
\end{align*}
where $n$ is the sequence length.
But, since  $(X_{\ell 1}^u)^2=X_{\ell 1}^u$, 
\begin{align*}
\EE\left ( \left ( \alpha  {X_{1\ell}^u} +\beta  {X_{2\ell}^v }+\gamma  {X_{3\ell}^w} \right )^3 \right )
&= \EE  ( \alpha^3 {X_{11}^u }+\beta^3  {X_{21}^v  } +\gamma^3  {X_{31}^w }\\
& \ \ \ \ \ \  +3(\alpha^2 \beta+\alpha\beta^2)  {X_{11}^uX_{21}^v}+3(\alpha^2 \gamma+\alpha \gamma^2) {X_{11}^uX_{31}^w}\\& \ \  \ \ \ \   +3(\beta^2 \gamma +\beta\gamma^2) {X_{21}^vX_{31}^w}
+6\alpha\beta\gamma X_{11}^uX_{21}^vX_{31}^w)\\
&= \alpha^3 \pi_1^u +\beta^3 \pi_2^v +\gamma^3 \pi_3^w -3\alpha \beta (\alpha+\beta) \EE( {X_{11}^uX_{21}^v})\\
&\ \ \ \ \ \ \ -3\alpha \gamma (\alpha+\gamma) \EE( {X_{11}^uX_{31}^w})-3 \beta \gamma(\beta+\gamma) \EE( {X_{21}^vX_{31}^w}))\\
&\ \ \ \ \ \ \ +6\alpha\beta\gamma \EE(X_{11}^uX_{21}^vX_{31}^w).
\end{align*}

Using $\EE (X_{11}^uX_{21}^v)=P_{uv+}$ and variants, and $\EE(X_{11}^uX_{21}^vX_{31}^w)=P_{uvw}$,
this shows
\begin{align*}
\EE(Y_{uvw}^3 )&=n \big( \alpha^3 \pi_1^u +\beta^3 \pi_2^v+\gamma^3 \pi_3^w -3\alpha \beta (\alpha+\beta) P_{uv+} -3\alpha \gamma (\alpha+\gamma) P_{u+w}\\
&\ \ \ \ \ \ \ -3 \beta \gamma(\beta+\gamma)P_{+vw}+ 6\alpha\beta\gamma P_{uvw} \big),\end{align*}
and the claim readily follows.
\end{proof}

\end{document}